\documentclass[11pt, a4paper]{article}

\usepackage[utf8]{inputenc}
\usepackage[T1]{fontenc}
\usepackage[margin=1in]{geometry} 
\usepackage{amsmath, amssymb, amsthm, amsfonts} 
\usepackage{mathtools}          
\usepackage{natbib}             
\usepackage{graphicx}           
\usepackage{booktabs}
\usepackage{float}              
\usepackage{placeins}
\usepackage{algorithm}          
\usepackage{algorithmic}        
\usepackage{bm}      
\usepackage{rotating}
\usepackage{setspace}           
\usepackage{lineno}             
\usepackage{color}              
\usepackage{multirow}           
\usepackage{caption}            
\usepackage{subcaption}         
\usepackage{authblk}            
\usepackage{url}    
\usepackage{pgfplots}

\usepackage{xcolor}

\usepackage[colorlinks=true, citecolor=blue, linkcolor=blue, urlcolor=cyan]{hyperref}
\newcommand{\argmin}{\mathop{\mathrm{argmin}}}

\newtheorem{theorem}{Theorem}
\newtheorem{lemma}{Lemma}
\newtheorem{proposition}{Proposition}
\newtheorem{corollary}{Corollary}
\newtheorem{assumption}{Assumption}
\newtheorem{remark}{Remark}

\title{\textbf{Group-Sparse Smoothing for Longitudinal Models with Time-Varying Coefficients}}

\author[1]{Yu Lu}
\author[1]{Tianni Zhang}
\author[1]{Yuyao Wang}
\author[1]{Mengfei Ran\footnote{Corresponding author. \url{mengfei.ran@xjtlu.edu.cn}}}
\affil[1]{Academy of Life and Natural Sciences, Xi'an Jiaotong-Liverpool University}

\date{}

\begin{document}

\maketitle

\begin{abstract}
Longitudinal associations may vary over time, yet allowing every regression effect to be dynamic can inflate estimation variance and obscure interpretable structure. We develop time-varying-effect selection (TV-Select), a group-sparse smoothing framework that classifies covariate effects as zero, constant, or time varying. Each coefficient is decomposed into a constant mean and a centered temporal deviation represented by a full-rank, L2-normalized effective spline basis. A group penalty identifies varying components, while a roughness penalty controls their curvature. The resulting convex criterion is solved by cyclic block proximal-gradient updates and followed by smooth refitting. Under a full-column-rank unpenalized design and an effective model dimension that is small relative to the total number of observations, we establish prediction and parameter rates, blockwise function-estimation bounds, and exact recovery of the varying set under irrepresentability and beta-min conditions. A stable classification refit further separates zero from constant effects. For fixed-dimensional contrasts, we construct an oracle-equivalent one-step estimator with cluster-robust asymptotic normality and consistent sandwich variance estimation. Simulations demonstrate that TV-Select combines low false-positive rates with accurate function estimation and competitive prediction across a range of longitudinal settings. An application to Sleep-EDF data produces smooth and parsimonious temporal effect estimates with essentially unchanged held-out predictive performance.

\vspace{0.5cm}
\end{abstract}

\section{Introduction}\label{sec:1}

Longitudinal studies collect repeated measurements on the same subjects and are widely used in biomedical, public health, and social science research. \citet{diggle2002analysis} described general modeling principles for longitudinal outcomes, while \citet{laird1982random} established the random-effects framework underlying linear mixed models for subject-specific heterogeneity. From a marginal perspective, \citet{liang1986longitudinal} introduced generalized estimating equations (GEE) for population-level inference under working correlation structures.

Constant-effect regression cannot represent associations that evolve over time. Varying-coefficient models (VCMs) address this restriction by allowing regression effects to change smoothly with time \citep{hastie1993varying}. Their statistical foundations include estimation theory for VCMs \citep{fan1999statistical}, nonparametric smoothing for longitudinal observations \citep{hoover1998nonparametric}, and a general treatment of nonparametric longitudinal regression \citep{wu2006nonparametric}. In many applications, however, temporal flexibility is needed for only a subset of predictors; the remaining effects are constant or inactive.

Structural identification in longitudinal VCMs requires determining whether each covariate is ($i$) irrelevant, ($ii$) constant, or ($iii$) time-varying. Penalization provides a general mechanism for this task. The Lasso \citep{tibshirani1996regression} established a basic framework for sparse estimation, and the nonconcave penalization approach of \citet{fan2001variable} provided oracle-type results. For grouped structures, \citet{yuan2006model} introduced the Group-Lasso, while \citet{meier2008group} extended blockwise penalties to generalized models. In the varying-coefficient setting, \citet{wang2008variable} developed spline-based variable selection methods for nonparametric VCMs, and \citet{wei2011group} studied estimation and selection in high-dimensional VCMs. Related work on additive models \citep{huang2010variable} established componentwise sparsity results under basis expansions.

Smoothness control is also required for interpretable coefficient trajectories. Classical spline approximation theory is summarized by \citet{schumaker1981spline}. \citet{eilers1996flexible} combined B-splines with roughness penalties in P-splines, and \citet{antoniadis2012penalized} adapted this construction to variable selection in VCMs. Additive regression provides a broader structured nonparametric setting \citep{stone1985additive}, and component-selection methods show that sparsity and smoothing can be estimated jointly \citep{lin2006component}. A longitudinal VCM procedure needs to control structural sparsity and functional regularity within the same criterion.

Efficient estimation and variable selection have been studied for semiparametric varying-coefficient partially linear models \citep{kai2011semiparametric}, and consistent model selection has been established for marginal generalized additive models \citep{xue2010consistent}. High-dimensional longitudinal selection and inference have also been developed through penalized GEE and time-varying coefficient models \citep{wang2012pgee,chen2018inference}. The literature addresses three related objectives: structural recovery, smooth function estimation, and valid longitudinal inference.

The spline-selection procedures of \citet{wang2008variable} and \citet{antoniadis2012penalized} are the closest precedents for the present work. TV-Select differs by combining an explicit classification of zero, constant, and varying effects with a centered, full-rank, $L_2$-normalized deviation basis, distinct selection and smooth-refitting estimators, and cluster-robust post-selection inference. The earlier methods address related variable-selection and smoothing problems, but not this combination of structural classification, normalized effective coordinates, and oracle-equivalent cluster-robust inference.

TV-Select addresses structural identification in longitudinal VCMs. Each coefficient is decomposed into a constant mean and a centered deviation, and time variation is identified through nonzero deviation blocks on an $L_2$-normalized effective spline basis. The convex criterion combines block sparsity and roughness regularization and is solved by cyclic proximal-gradient updates with explicit Lipschitz step sizes and a KKT stopping rule. The theoretical results include estimation and blockwise function bounds, exact recovery of the time-varying set, rates for smooth refitting, consistent separation of zero and constant effects, and cluster-robust post-selection inference through an oracle-equivalent one-step estimator. The simulation design uses a common data-generating mechanism, preprocessing rule, tuning protocol, and set of evaluation metrics across methods. Because the $p$ constant components are unpenalized, the present theory requires a full-column-rank $\bar{\mathbf X}$ and $p+s_vq_e=o(n_\bullet)$. When baseline covariates are time invariant within subject, $\operatorname{rank}(\bar{\mathbf X})\le N$, so the scope additionally requires $p<N$ and does not include the ultra-high-dimensional regime.

The rest of the paper is organized as follows. Section~\ref{sec:2} introduces the model, the structural decomposition, and the penalized estimator. Section~\ref{sec:3} gives the block proximal-gradient algorithm, and Section~\ref{sec:4} develops the asymptotic results. Sections~\ref{sec:5} and \ref{sec:6} contain the simulation study and Sleep-EDF application, respectively. Section~\ref{sec:7} concludes with discussion and extensions.

\section{Methodology}\label{sec:2}

\subsection{Longitudinal Model}

Consider a longitudinal study with $N$ subjects. For subject $i$ ($i=1,\dots,N$), we observe a response $y_{ij}$ and a $p$-dimensional covariate vector $\mathbf{x}_{ij}=(x_{ij1},\dots,x_{ijp})^\top$ at time $t_{ij}$, $j=1,\dots,n_i$. Let $n_\bullet=\sum_{i=1}^N n_i$ denote the total number of observations. Without loss of generality, the observation times are rescaled to $t_{ij}\in[0,1]$.

Conditional on the observed design and times, we assume the working mean model
\begin{equation}\label{eq:model}
y_{ij}
=
\beta_0
+\sum_{k=1}^p x_{ijk}\beta_k(t_{ij})
+\varepsilon_{ij},
\end{equation}
where $\beta_k(\cdot)$ are unknown smooth coefficient functions and $\mathbb{E}(\varepsilon_{ij}\mid\mathcal D)=0$, with $\mathcal D$ denoting the full design and observation schedule. Dependence among observations from the same subject is allowed; subjects are the independent sampling units in the theory. Equation~\eqref{eq:model} specifies the conditional mean and does not impose working independence. The least-squares criterion uses working independence, while score concentration and cluster-robust variance retain within-subject dependence, consistent with the distinction between mean and association models in longitudinal analysis \citep{diggle2002analysis,liang1986longitudinal}. Explicit covariance modeling can improve efficiency but is outside the present selection criterion \citep{bai2023nvcssl}.

For each $k=1,\dots,p$, define
\[
\mu_k=\int_0^1\beta_k(u)\,du
\]
and decompose
\begin{equation*}
\beta_k(t)=\mu_k+g_k(t),
\qquad
\int_0^1 g_k(t)\,dt=0.
\end{equation*}
Here $\mu_k$ is the uniform-time average of the $k$th coefficient function over the normalized interval $[0,1]$, and $g_k(t)$ is its centered temporal deviation. The definition uses Lebesgue measure; for irregular visits it need not equal an observation-time-weighted average, and it has no causal interpretation without additional assumptions. The resulting structural partition is
\[
\mathcal{S}_{\mathrm{zero}}
=\{k:\mu_k=0,\ g_k(t)\equiv0\},\quad
\mathcal{S}_{\mathrm{const}}
=\{k:\mu_k\neq0,\ g_k(t)\equiv0\},\quad
\mathcal{S}_{\mathrm{vary}}
=\{k:g_k(t)\not\equiv0\}.
\]
The goal of TV-Select is to recover this structure while estimating the nonzero coefficient functions.

\subsection{Spline Approximation and Effective Subspace}

Let $\boldsymbol{B}(t)=(B_1(t),\dots,B_q(t))^\top$ be a degree-$d$ B-spline basis on $[0,1]$ \citep{schumaker1981spline}. Define $\bar{\boldsymbol B}=\int_0^1\boldsymbol B(u)\,du$ and
\[
\tilde{\boldsymbol B}(t)=\boldsymbol B(t)-\bar{\boldsymbol B}.
\]
Then $\int_0^1\tilde{\boldsymbol B}(t)\,dt=0$.  Because B-splines form a partition of unity, the centered raw basis has a deterministic null direction. Define its population Gram matrix
\[
G_B=\int_0^1\tilde{\boldsymbol B}(t)
\tilde{\boldsymbol B}(t)^\top dt
=U_e\Lambda_eU_e^\top,
\]
where $\Lambda_e$ contains the strictly positive eigenvalues and $U_e\in\mathbb R^{q\times q_e}$ the associated orthonormal eigenvectors. For ordinary centered B-splines, $q_e=q-1$.  We use the normalized effective basis
\begin{equation}\label{eq:normalized_basis}
\boldsymbol C(t)=\Lambda_e^{-1/2}U_e^\top
\tilde{\boldsymbol B}(t)\in\mathbb R^{q_e}.
\end{equation}
It satisfies
\[
\int_0^1\{\boldsymbol C(t)^\top a\}^2dt=\|a\|_2^2,
\qquad a\in\mathbb R^{q_e}.
\]
Under this normalization, the group norm is exactly the $L_2$ norm of the represented centered function. In particular, $\boldsymbol\theta_k\in\mathbb R^{q_e}$ is unique and
\[
g_k(t)\approx \boldsymbol C(t)^\top\boldsymbol\theta_k.
\]
In computation, $G_B$ is evaluated by high-order numerical quadrature and only eigenvalues exceeding a fixed relative rank tolerance are retained. This implementation of \eqref{eq:normalized_basis} removes the deterministic null direction before estimation rather than imposing an arbitrary coefficient constraint. A Moore--Penrose solution or a ridge solve $(A+\rho I)^{-1}$ may be used only in a separately labeled numerical sensitivity analysis for an otherwise singular auxiliary fit. Neither is substituted for the full-rank selection, classification, or inferential estimators analyzed below, and a ridge solve is not treated as algebraically equivalent to $A^\dagger$.

For each $k$, define the $n_\bullet\times q_e$ matrix $\mathbf Z_k$ with rows $x_{ijk}\boldsymbol C(t_{ij})^\top$. Let $\mathbf X$ be the stacked covariate matrix. The approximating model can be written as
\[
\mathbf y
\approx
\beta_0\mathbf 1_{n_\bullet}
+\mathbf X\boldsymbol\mu
+\sum_{k=1}^p\mathbf Z_k\boldsymbol\theta_k
+\boldsymbol\varepsilon.
\]
If $\Omega_B$ is the raw-basis roughness matrix and $T=U_e\Lambda_e^{-1/2}$, then the corresponding matrix in normalized coordinates is $\Omega=T^\top\Omega_BT\succeq0$. We rescale it so $\|\Omega\|_{\rm op}\le1$; this changes only the numerical scale of $\lambda_2$. The implementation uses the P-spline difference penalty $\Omega_B=D_2^\top D_2$ \citep{eilers1996flexible}, whereas an integrated squared second-derivative matrix may be used without changing the arguments below.

\subsection{Penalized Selection Estimator}

For tuning parameters $\lambda_1>0$ and $\lambda_2\ge0$, we estimate $\Theta=(\beta_0,\boldsymbol\mu^\top,\boldsymbol\theta_1^\top,\dots, \boldsymbol\theta_p^\top)^\top$ by minimizing the scaled criterion
\begin{equation}\label{eq:objective}
\widehat\Theta
=
\argmin_{\Theta}
\left[
\frac{1}{2n_\bullet}
\left\|\mathbf y-\beta_0\mathbf 1_{n_\bullet}-\mathbf X\boldsymbol\mu-\sum_{k=1}^p\mathbf Z_k\boldsymbol\theta_k\right\|_2^2
+\sum_{k=1}^p
\left\{
\lambda_1\|\boldsymbol\theta_k\|_2+\lambda_2\boldsymbol\theta_k^\top\Omega\boldsymbol\theta_k
\right\}
\right].
\end{equation}
The group penalty $\lambda_1\|\boldsymbol\theta_k\|_2$ selects time-varying effects by shrinking entire deviation blocks to zero, following the group-sparsity principle \citep{yuan2006model,lounici2011oracle}. The roughness penalty $\lambda_2\boldsymbol\theta_k^\top\Omega\boldsymbol\theta_k$ stabilizes the selected curves, as in penalized spline selection for VCMs \citep{antoniadis2012penalized}. The time-varying set is estimated by
\[
\widehat{\mathcal S}_{\mathrm{vary}}
=\{k:\|\widehat{\boldsymbol\theta}_k\|_2>0\}.
\]
The penalized estimator in \eqref{eq:objective} is used only for structural selection.

\subsection{Refitting and Classification}

Coefficient curves and predictions are constructed from a post-selection smooth refit that removes group-penalty shrinkage while retaining the selected amount of smoothness. Separating selection from refitting follows the bias-reduction principle underlying oracle and adaptive penalization \citep{fan2001variable,zou2006adaptive,wang2008adaptivegroup}:
\begin{equation}\label{eq:smooth_refit}
(\hat\gamma^{\,F},\hat{\boldsymbol\theta}^{\,F}_{\widehat S})
=
\argmin_{\gamma,\boldsymbol\theta_{\widehat S}}
\left\{
\frac{1}{2n_\bullet}
\|\mathbf y-\bar{\mathbf X}\gamma-
\mathbf Z_{\widehat S}\boldsymbol\theta_{\widehat S}\|_2^2
+\hat\lambda_2
\sum_{k\in\widehat S}
\boldsymbol\theta_k^\top\Omega\boldsymbol\theta_k
\right\},
\qquad
\widehat S=\widehat{\mathcal S}_{\mathrm{vary}}.
\end{equation}
Equation~\eqref{eq:smooth_refit} removes the group penalty only after selection and retains the EBIC-selected roughness penalty. The resulting estimator is distinct from a fully unpenalized refit.

The distinction between constant and varying components is also central to structure specification in semi-varying coefficient and panel models \citep{li2015gsvcm,ke2016structure}. TV-Select classifies this structure through an additional unpenalized refit. Specifically, for a candidate set $A\subset\{1,\ldots,p\}$, let
\begin{equation}\label{eq:structure_refit}
(\hat\gamma^{\,R}(A),\hat{\boldsymbol\theta}^{\,R}_A)
=
\argmin_{\gamma,\boldsymbol\theta_A}
\frac{1}{2n_\bullet}
\|\mathbf y-\bar{\mathbf X}\gamma-\mathbf Z_A\boldsymbol\theta_A\|_2^2,
\qquad
\bar{\mathbf X}=[\mathbf 1,\mathbf X],
\end{equation}
where $\gamma=(\beta_0,\boldsymbol\mu^\top)^\top$. The classification refit is performed only if $D_A^{\rm all}=[\bar{\mathbf X},\mathbf Z_A]$ has full numerical column rank under a prespecified tolerance. If rank deficiency is detected, the classification is flagged as non-identifiable rather than replaced by a nonunique generalized-inverse solution. A ridge-stabilized refit may be examined separately in a sensitivity analysis, but it is a different estimator and is not covered by Corollary~\ref{cor:mu_threshold}. Write $\hat{\boldsymbol\mu}^{\,R}$ for the constant-effect coordinates from \eqref{eq:structure_refit} with $A=\widehat{\mathcal S}_{\mathrm{vary}}$. Zero and constant effects among the non-varying variables are separated by thresholding this refitted estimator:
\[
\widehat{\mathcal S}_{\mathrm{const}}
=
\{k\notin\widehat{\mathcal S}_{\mathrm{vary}}:|\hat\mu^{\,R}_k|>\tau_N\},\qquad
\widehat{\mathcal S}_{\mathrm{zero}}
=
\{k\notin\widehat{\mathcal S}_{\mathrm{vary}}:|\hat\mu^{\,R}_k|\le\tau_N\},
\]
where the precise order of $\tau_N$ is given in Corollary~\ref{cor:mu_threshold}. The refit is used only after structural selection and is not used to construct coefficient curves or predictions.

\subsection{Post-selection Inference}

Inference for longitudinal and high-dimensional VCMs requires accounting for both temporal nuisance functions and within-subject dependence \citep{chen2018inference,hu2021robust,dai2021vcqr}. The inferential estimator in Theorem~\ref{thm:T3_oracle} can be written either as a post-selection oracle refit or as an exactly equivalent one-step correction. It is not used for prediction or curve estimation. After obtaining $\widehat S=\widehat{\mathcal S}_{\mathrm{vary}}$ and $\widehat S_c=\widehat{\mathcal S}_{\mathrm{const}}$, define the inferential post-selection design
\[
\widehat D^{\,I}
=
[\mathbf1,\mathbf X_{\widehat S_c},\mathbf X_{\widehat S},
\mathbf Z_{\widehat S}]
\]
and refit all of its coefficients by ordinary least squares.  We denote the constant-effect coordinates of this second refit by $\hat{\boldsymbol\mu}^{\,PS}_{\widehat S_c}$.  Equivalently, let $\hat W=[\mathbf1,\mathbf X_{\widehat S},\mathbf Z_{\widehat S}]$, $\tilde{\mathbf X}_{\widehat S_c}=M_{\hat W}\mathbf X_{\widehat S_c}$, and $\hat A_{\widehat S_c}=\tilde{\mathbf X}_{\widehat S_c}^\top\tilde{\mathbf X}_{\widehat S_c}/n_\bullet$.  Starting from any coefficients $(\hat{\boldsymbol\mu}^{\,P}_{\widehat S_c},\hat\eta^P)$ in the selected model, define
\begin{equation}\label{eq:one_step_definition}
\hat{\boldsymbol\mu}^{\,DB}_{\widehat S_c}
=\hat{\boldsymbol\mu}^{\,P}_{\widehat S_c}
+\hat A_{\widehat S_c}^{-1}
\frac{\tilde{\mathbf X}_{\widehat S_c}^\top
\{\mathbf y-\mathbf X_{\widehat S_c}
\hat{\boldsymbol\mu}^{\,P}_{\widehat S_c}-\hat W\hat\eta^P\}}
{n_\bullet}.
\end{equation}
Because $\tilde{\mathbf X}_{\widehat S_c}^\top\hat W=0$, the starting values cancel and \eqref{eq:one_step_definition} is exactly the Frisch--Waugh--Lovell coefficient from the post-selection least-squares refit. The refit in \eqref{eq:structure_refit} retains all columns of $\mathbf X$ and is used to classify non-varying variables. The lower-dimensional design $\widehat D^{\,I}$ is used for post-selection inference, with $\hat{\boldsymbol\mu}^{\,DB}_{\widehat S_c}=\hat{\boldsymbol\mu}^{\,PS}_{\widehat S_c}$ whenever the displayed inverse exists.

The four estimators have distinct roles:
\begin{center}
\small
\begin{tabular}{lll}
\toprule
Estimator & Defining equation & Role\\
\midrule
$\widehat\Theta$ & \eqref{eq:objective} & select
$\widehat{\mathcal S}_{\mathrm{vary}}$\\
$(\hat\gamma^{\,F},\hat{\boldsymbol\theta}^{\,F})$
& \eqref{eq:smooth_refit} & estimate curves and predict\\
$(\hat\gamma^{\,R},\hat{\boldsymbol\theta}^{\,R})$
& \eqref{eq:structure_refit} & classify constant and zero effects\\
$\hat{\boldsymbol\mu}^{\,PS}
=\hat{\boldsymbol\mu}^{\,DB}$ & \eqref{eq:one_step_definition}
& post-selection inference\\
\bottomrule
\end{tabular}
\end{center}

\section{Algorithm}\label{sec:3}

The criterion in \eqref{eq:objective} is minimized by cyclic block proximal gradient, a block-separable optimization strategy related to coordinate methods for nonsmooth objectives \citep{tseng2001convergence}. The unpenalized block $(\beta_0,\boldsymbol\mu)$ is updated by least squares conditional on the current varying blocks, and each $\boldsymbol\theta_k$ is updated by one proximal-gradient step.

\subsection{Block Updates}

Given current $\{\boldsymbol\theta_k\}$, the intercept and constant effects are updated jointly by
\begin{equation}\label{eq:update_mu_joint}
\gamma=
\begin{pmatrix}\beta_0\\ \boldsymbol\mu\end{pmatrix}
\leftarrow
\left(\frac{1}{n_\bullet}\bar{\mathbf X}^\top\bar{\mathbf X}\right)^{-1}
\frac{1}{n_\bullet}\bar{\mathbf X}^\top
\left(\mathbf y-\sum_{k=1}^p\mathbf Z_k\boldsymbol\theta_k\right),
\end{equation}
which is well defined under the full-rank condition used in Proposition~\ref{prop:algorithm_convergence} and Assumption~\ref{ass:A1}.  Equivalently, this update is the least-squares fit of $\mathbf y-\sum_k\mathbf Z_k\boldsymbol\theta_k$ on $\bar{\mathbf X}$. A numerical sensitivity analysis may instead use $\{\bar{\mathbf X}^{\top}\bar{\mathbf X}/n_\bullet+\rho_\mu R_\mu\}^{-1}$, where $R_\mu=\operatorname{diag}(0,1,\ldots,1)$ and $\rho_\mu>0$, but that ridge-perturbed estimator is not the target of \eqref{eq:objective} or the theory below.

For the $k$th time-varying block, define the partial residual
\[
\mathbf r^{(k)}
=
\mathbf y-\beta_0\mathbf 1_{n_\bullet}-\mathbf X\boldsymbol\mu
-\sum_{\ell\neq k}\mathbf Z_\ell\boldsymbol\theta_\ell.
\]
The smooth part of the block objective is
\[
\ell_k(\boldsymbol\theta)
=
\frac{1}{2n_\bullet}\|\mathbf r^{(k)}-\mathbf Z_k\boldsymbol\theta\|_2^2
+\lambda_2\boldsymbol\theta^\top\Omega\boldsymbol\theta,
\]
with gradient
\begin{equation*}
\nabla\ell_k(\boldsymbol\theta)
=
-\frac{1}{n_\bullet}\mathbf Z_k^\top(\mathbf r^{(k)}-\mathbf Z_k\boldsymbol\theta)
+2\lambda_2\Omega\boldsymbol\theta.
\end{equation*}
A valid Lipschitz constant for this gradient is
\begin{equation}\label{eq:lipschitz}
L_k
=
\lambda_{\max}\!\left(\frac{1}{n_\bullet}\mathbf Z_k^\top\mathbf Z_k+2\lambda_2\Omega\right).
\end{equation}
If $L_k=0$, the block is identically uninformative and is set to zero. Otherwise, with $0<\alpha_k\le L_k^{-1}$, the proximal-gradient update is
\begin{equation}\label{eq:pg_update}
\mathbf v_k
=
\boldsymbol\theta_k-\alpha_k\nabla\ell_k(\boldsymbol\theta_k),
\qquad
\boldsymbol\theta_k
\leftarrow
\left(1-\frac{\alpha_k\lambda_1}{\|\mathbf v_k\|_2}\right)_+\mathbf v_k.
\end{equation}
This update does not increase the block objective. For the sufficient-descent
and convergence result in Proposition~\ref{prop:algorithm_convergence}, we use $\alpha_k\in[\underline{\alpha},(1-\nu)L_k^{-1}]$ for some \(\nu\in(0,1)\). The implementation uses \(\alpha_k=0.99L_k^{-1}\), subject to the stated numerical cap. The expression is the exact proximal-gradient update for the block objective and guarantees descent for $\alpha_k\le L_k^{-1}$.

To make convergence numerically verifiable, let
\[
\mathbf e=\mathbf y-\bar{\mathbf X}\gamma
-\sum_{k=1}^p\mathbf Z_k\boldsymbol\theta_k
\]
at a coherent sweep-end iterate and define
\[
R_\gamma=
\left\|\frac{\bar{\mathbf X}^{\top}\mathbf e}{n_\bullet}\right\|_\infty,
\qquad
R_k=
\begin{cases}
\left\|-\dfrac{\mathbf Z_k^\top\mathbf e}{n_\bullet}
+2\lambda_2\Omega\boldsymbol\theta_k
+\lambda_1\dfrac{\boldsymbol\theta_k}{\|\boldsymbol\theta_k\|_2}
\right\|_2,
& \boldsymbol\theta_k\ne0,\\[10pt]
\left(\left\|\dfrac{\mathbf Z_k^\top\mathbf e}{n_\bullet}\right\|_2
-\lambda_1\right)_+,
& \boldsymbol\theta_k=0.
\end{cases}
\]
The KKT residual is $R_{\mathrm{KKT}}=\max\{R_\gamma,R_1,\ldots,R_p\}$. The zero-block expression is the distance from the negative smooth gradient to the group-penalty subdifferential and therefore remains well defined at the origin.

\begin{algorithm}[H]
\caption{Block Coordinate Proximal-Gradient Algorithm}
\label{alg:bcd}
\begin{algorithmic}[1]
\STATE \textbf{Input:} Data $\{y_{ij},\mathbf x_{ij},t_{ij}\}$; normalized effective basis $\boldsymbol C(\cdot)$; penalties $\lambda_1,\lambda_2$; tolerances $\epsilon_{\rm obj},\epsilon_{\rm KKT}$.
\STATE Construct $\mathbf y$, $\mathbf X$, and $\{\mathbf Z_k\}_{k=1}^p$.
\STATE Initialize $\boldsymbol\theta_k^{(0)}=\mathbf0$ and initialize $(\beta_0^{(0)},\boldsymbol\mu^{(0)})$ by the constant-effect fit.
\FOR{$t=0,1,2,\ldots$}
\STATE Update $(\beta_0,\boldsymbol\mu)$ jointly by \eqref{eq:update_mu_joint}.
\FOR{$k=1,\ldots,p$}
\STATE Form $\mathbf r^{(k)}$ and compute $L_k$ in \eqref{eq:lipschitz}.
\STATE Update $\boldsymbol\theta_k$ by the proximal-gradient step \eqref{eq:pg_update}.
\ENDFOR
\STATE Recompute $\gamma$ by \eqref{eq:update_mu_joint} so the iterate is coherent.
\STATE Stop if both the relative objective decrease is below $\epsilon_{\rm obj}$ and $R_{\mathrm{KKT}}\le\epsilon_{\rm KKT}$.
\ENDFOR
\STATE Form $\widehat S=\{k:\|\hat{\boldsymbol\theta}_k\|_2>0\}$ and compute the smooth refit \eqref{eq:smooth_refit}.
\STATE \textbf{Output:} $\widehat S$ from the penalized stage and the smooth-refit curves
$\hat\beta_k^{\,F}(t)=\hat\mu_k^{\,F}+
\boldsymbol C(t)^\top\hat{\boldsymbol\theta}_k^{\,F}$.
\end{algorithmic}
\end{algorithm}

\begin{proposition}[Convergence of Algorithm~\ref{alg:bcd}]
\label{prop:algorithm_convergence}
Suppose $\lambda_1>0$, $\bar{\mathbf X}$ has full column rank, and the initial objective level set is bounded in the normalized effective coordinates.  For every informative block, use $\alpha_k\in[\underline\alpha,(1-\nu)L_k^{-1}]$ for constants $\underline\alpha>0$ and $\nu\in(0,1)$; an uninformative block with $L_k=0$ is fixed at zero.  Then the objective values generated by Algorithm~\ref{alg:bcd} decrease monotonically, the block increments converge to zero, and every accumulation point is a global minimizer of \eqref{eq:objective}. The fitted values converge.  If the minimizer is unique in the normalized effective coordinates, the entire coefficient sequence converges to it.
\end{proposition}

\section{Asymptotic Properties}\label{sec:4}

All spline coordinates in this section are the normalized effective coordinates of Section~\ref{sec:2}, for which $\|\boldsymbol C(\cdot)^\top a\|_{L_2}=\|a\|_2$ exactly. Put $n=n_\bullet$, $\bar{\mathbf X}=[\mathbf 1,\mathbf X]$, $P_X=P_{\bar{\mathbf X}}$, $M_X=I-P_X$, $\tilde{\mathbf Z}_k=M_X\mathbf Z_k$, and $\tilde{\mathbf Z}=[\tilde{\mathbf Z}_1,\ldots,\tilde{\mathbf Z}_p]$. For a block vector $v=(v_1^\top,\ldots,v_p^\top)^\top$, define
\[
\|v\|_{2,1}=\sum_{k=1}^p\|v_k\|_2,\qquad
\|v\|_{2,\infty}=\max_{1\le k\le p}\|v_k\|_2.
\]
The mixed norms are standard in oracle analyses of group-sparse estimators \citep{lounici2011oracle}.
Let $S=\mathcal S_{\mathrm{vary}}$, $s_v=|S|$, and $\mathcal C(S,3)=\{v:\|v_{S^c}\|_{2,1}\le3\|v_S\|_{2,1}\}$. Because $\boldsymbol\mu$ is unpenalized, the theory requires $\bar{\mathbf X}$ to have full column rank and $p+s_vq_e=o(n)$.  If the baseline covariates are time invariant within subject, then $\operatorname{rank}(\bar{\mathbf X})\le N$ and full column rank additionally requires $p<N$.  The condition $p<n$ alone is insufficient, and the present theory does not cover $p\gg n$.

\subsection{Assumptions}

\begin{assumption}[Independent subjects and clustered concentration]
\label{ass:A1}
Conditional on the design and observation times, the subject vectors $\boldsymbol\varepsilon_i=(\varepsilon_{i1},\ldots,\varepsilon_{in_i})^\top$ are independent, mean zero, and satisfy
\[
\|\langle a,\boldsymbol\varepsilon_i\rangle\|_{\psi_2}
\le K_\varepsilon\|a\|_2
\quad\text{for every deterministic }a\in\mathbb R^{n_i}.
\]
Within-subject dependence is unrestricted subject to this bound, and $\max_i n_i\le m_0<\infty$, so $n\asymp N$.  The eigenvalues of $n^{-1}\bar{\mathbf X}^{\top}\bar{\mathbf X}$ lie in $[\kappa_X,\kappa_X^{-1}]$, and
\[
\max_{k\le p}\lambda_{\max}
\left(n^{-1}\tilde{\mathbf Z}_k^\top\tilde{\mathbf Z}_k\right)
\le K_Z
\]
with probability tending to one.  The design is fixed or independent of the errors.
\end{assumption}

\begin{assumption}[Spline approximation]\label{ass:A2}
For every $k\in S$, $g_{0k}$ belongs to a Hölder ball of order $2<r\le d+1$, where $d$ is the spline degree. There is a unique normalized effective coefficient $\boldsymbol\theta_{0k}^\ast\in\mathbb R^{q_e}$ such that
\[
\sup_{t\in[0,1]}
\left|g_{0k}(t)-\boldsymbol C(t)^\top\boldsymbol\theta_{0k}^\ast\right|
\le K_gq_e^{-r}.
\]
Set $\boldsymbol\theta_{0k}^\ast=0$ for $k\notin S$. The stacked approximation remainder includes the covariates explicitly:
\[
r_{ij}
=
\sum_{k\in S}x_{ijk}
\{g_{0k}(t_{ij})
-\boldsymbol C(t_{ij})^\top\boldsymbol\theta_{0k}^\ast\},
\qquad
\mathbf r=(r_{11},\ldots,r_{Nn_N})^\top,
\]
where the stacking order is the same as for $\mathbf y$.  We assume
\[
\frac{\|\mathbf r\|_2^2}{n}=O(s_vq_e^{-2r}),\qquad
b_n:=
\left\|\frac{\tilde{\mathbf Z}^\top M_X\mathbf r}{n}\right\|_{2,\infty}.
\]
\end{assumption}

\begin{assumption}[Compatibility, joint identifiability, and active inverse]
\label{ass:A3}
There are constants $\kappa_Z,\kappa_D,C_H,C_G>0$ such that, with probability tending to one, for every $v\in\mathcal C(S,3)$ and every $a\in\mathbb R^{p+1}$,
\begin{align}
\frac{\|\tilde{\mathbf Z}v\|_2^2}{n}
&\ge \kappa_Z^2\|v_S\|_2^2,\label{eq:compatibility}\\
\frac{\|\bar{\mathbf X}a+\mathbf Zv\|_2^2}{n}
&\ge \kappa_D^2\{\|a\|_2^2+\|v_S\|_2^2\}.
\label{eq:joint_identifiability}
\end{align}
Let $\Omega_p=I_p\otimes\Omega$,
\[
H=\frac{\tilde{\mathbf Z}^\top\tilde{\mathbf Z}}{n}
+2\lambda_2\Omega_p,\qquad G_S=H_{SS}.
\]
The matrix $G_S$ is positive definite and its inverse is stable in block maximum norm:
\begin{equation}\label{eq:block_inverse_condition}
\|G_S^{-1}w_S\|_{2,\infty}\le C_G\|w_S\|_{2,\infty}
\qquad\text{for every block vector }w_S.
\end{equation}
Moreover, $\max_{j,k}\|n^{-1}\tilde{\mathbf Z}_j^\top\tilde{\mathbf Z}_k\|_{\rm op}
\le C_H$. The latter condition controls leakage from erroneously nonzero inactive blocks in the blockwise KKT bound.
\end{assumption}

\begin{assumption}[Growth and tuning]\label{ass:A4}
Let
\[
\zeta_n=\sqrt{\frac{q_e+\log p}{n}},\qquad
\delta_n=\max_{k\in S}2\lambda_2
\|\Omega\boldsymbol\theta_{0k}^\ast\|_2.
\]
We assume $q_e+\log p=o(n)$, $p+s_vq_e=o(n)$, $p/n+s_v\lambda_1^2+s_vq_e^{-2r}=o(1)$, and
\[
\lambda_1=A_n\zeta_n,\qquad
b_n+\delta_n\le c_0\lambda_1,
\]
where $A_n\to\infty$ sufficiently slowly, $\lambda_1\to0$, and $c_0>0$ is sufficiently small. The roughness matrix satisfies $\|\Omega\|_{\rm op}=O(1)$ in normalized effective coordinates. For exact support recovery we strengthen the last display to $b_n+\delta_n=o(\lambda_1)$.
\end{assumption}

\begin{assumption}[Group irrepresentability]\label{ass:A5}
For some $\eta\in(0,1)$,
\begin{equation}\label{eq:irrepresentable}
\max_{k\notin S}\ 
\sup_{\|u\|_{2,\infty}\le1}
\left\|H_{kS}G_S^{-1}u\right\|_2\le1-\eta.
\end{equation}
\end{assumption}

\begin{assumption}[Classification-refit identifiability]
\label{ass:A6}
For classification after varying-set recovery, let $D_S^{\rm all}=[\bar{\mathbf X},\mathbf Z_S]$ and $B_S=\{(D_S^{\rm all})^\top D_S^{\rm all}/n\}^{-1}$. There are constants $\kappa_R,C_R>0$ such that the eigenvalues of the Gram matrix lie in $[\kappa_R,\kappa_R^{-1}]$ and
\[
\max_{\ell\in\mathcal I_\mu}\|e_\ell^\top B_S\|_1\le C_R,
\]
where $\mathcal I_\mu$ indexes the $p$ constant-effect coordinates. The conditions concern the population sequence of selected designs; numerical rank tolerances used to implement the refit are not asymptotic assumptions.
\end{assumption}

\begin{lemma}[Coordinatewise rate of the classification refit]
\label{lem:classification_refit_rate}
Under Assumptions~\ref{ass:A1}, \ref{ass:A2}, and \ref{ass:A6}, suppose
$s_v=O(1)$ and
$\log\{p+1+s_vq_e\}=O(\log p)$.  Then
\begin{align}
\left\|\frac{(D_S^{\rm all})^\top
\boldsymbol\varepsilon}{n}\right\|_\infty
&=O_p\!\left(\sqrt{\frac{\log p}{n}}\right),\label{eq:refit_score_bound}\\
\left\|\frac{(D_S^{\rm all})^\top\mathbf r}{n}\right\|_\infty
&=O(q_e^{-r}),\label{eq:refit_remainder_bound}
\end{align}
and the oracle classification refit satisfies
\begin{equation}\label{eq:refit_stability}
\|\hat{\boldsymbol\mu}^{\,R}-\boldsymbol\mu_0\|_\infty
=O_p(r_{\mu,n}),
\qquad
r_{\mu,n}=\sqrt{\frac{\log p}{n}}+q_e^{-r}.
\end{equation}
\end{lemma}

Lemma~\ref{lem:classification_refit_rate} isolates the additional conditions needed to separate zero and constant effects. Theorem~\ref{thm:T1_rate} controls the penalized varying blocks, but does not by itself imply a coordinatewise rate for the unpenalized classification refit. That rate follows only after combining clustered score concentration, the explicit approximation remainder, and the inverse-Gram row bound in Assumption~\ref{ass:A6}.

\begin{remark}[Scope of the conditions]
The penalty level is chosen to dominate stochastic score, approximation leakage, and first-order smoothing bias. Selection requires this margin, which is stronger than a tuning rule aimed only at mean-squared error. The unpenalized treatment of all $p$ constant effects also explains the restriction $p+s_vq_e=o(n)$. An ultra-high-dimensional extension would require an additional penalty on $\boldsymbol\mu$ and a different proof. The asymptotic statements treat the tuning sequences as deterministic, or as random sequences that satisfy the stated rate conditions with probability tending to one. A separate proof that EBIC selects such sequences is not asserted here.
\end{remark}

\subsection{Main Results}

\begin{theorem}[Estimation and blockwise function error]\label{thm:T1_rate}
Under Assumptions~\ref{ass:A1}--\ref{ass:A4}, every minimizer of
\eqref{eq:objective} satisfies
\begin{align}
\frac{\|\tilde{\mathbf Z}
(\hat{\boldsymbol\theta}-\boldsymbol\theta_0^\ast)\|_2^2}{n}
&=O_p(s_v\lambda_1^2),\label{eq:t1_profile_prediction}\\
\|(\hat{\boldsymbol\theta}-\boldsymbol\theta_0^\ast)_S\|_2^2
&=O_p(s_v\lambda_1^2),\qquad
\|(\hat{\boldsymbol\theta}-\boldsymbol\theta_0^\ast)_{S^c}\|_{2,1}
=O_p(s_v\lambda_1).\label{eq:t1_parameter}
\end{align}
Writing $\gamma_0=(\beta_0,\boldsymbol\mu_0^\top)^\top$ and $\hat\gamma=(\hat\beta_0,\hat{\boldsymbol\mu}^\top)^\top$,
\begin{align}
\frac{\|\bar{\mathbf X}(\hat\gamma-\gamma_0)
+\mathbf Z(\hat{\boldsymbol\theta}
-\boldsymbol\theta_0^\ast)\|_2^2}{n}
&=O_p\left(\frac pn+s_v\lambda_1^2+s_vq_e^{-2r}\right),\notag\\
\|\hat\gamma-\gamma_0\|_2^2+
\|(\hat{\boldsymbol\theta}-\boldsymbol\theta_0^\ast)_S\|_2^2
&=O_p\left(\frac pn+s_v\lambda_1^2+s_vq_e^{-2r}\right).
\label{eq:t1_joint_parameter}
\end{align}
If, in addition, $s_v=O(1)$, then the active KKT equations and \eqref{eq:block_inverse_condition} give
\begin{equation}\label{eq:t1_max}
\max_{k\in S}
\|\hat{\boldsymbol\theta}_k-\boldsymbol\theta_{0k}^\ast\|_2
=O_p(\lambda_1).
\end{equation}
Consequently,
\begin{equation}\label{eq:t1_function}
\max_{k\in S}\|\hat g_k-g_{0k}\|_{L_2}
=O_p(q_e^{-r}+\lambda_1).
\end{equation}
\end{theorem}

Theorem~\ref{thm:T1_rate} separates the two sources of function-estimation error: $q_e^{-r}$ is the spline approximation error, whereas $\lambda_1$ is the stochastic regularization error. The blockwise maximum bound is stronger than an aggregate prediction bound and shows that every active varying effect is estimated consistently when both terms vanish. The bound concerns the penalized estimator used for structural selection; the rate of the subsequent smooth refit is given below.

\begin{theorem}[Selection consistency for varying effects]
\label{thm:T2_select}
Assume that Assumptions~\ref{ass:A1}--\ref{ass:A5} hold, including the exact-support strengthening $b_n+\delta_n=o(\lambda_1)$ in Assumption~\ref{ass:A4},
$s_v=O(1)$, and
\[
\liminf_{n\to\infty}
\frac{\min_{k\in S}\|\boldsymbol\theta_{0k}^\ast\|_2}{\lambda_1}
>C_\star
\]
for a fixed $C_\star>C_G$, where $C_G$ is the block-inverse constant in \eqref{eq:block_inverse_condition}.  Then the profiled criterion has a unique solution and
\[
\Pr(\widehat{\mathcal S}_{\mathrm{vary}}=S)\longrightarrow1.
\]
\end{theorem}

Theorem~\ref{thm:T2_select} strengthens the estimation guarantees of Theorem~\ref{thm:T1_rate} to exact recovery of the varying-effect set. The beta-min condition keeps each truly varying block above the penalization level, while the diverging factor $A_n$ makes the stochastic block score negligible relative to $\lambda_1$ and the strict dual separation in Assumption~\ref{ass:A5} excludes inactive blocks. The conditions parallel the separation requirements used in group-sparse and high-dimensional varying-coefficient selection \citep{lounici2011oracle,wei2011group}, but are invoked here only for exact support recovery. Classification of the remaining effects as constant or zero is a separate step and is addressed by Corollary~\ref{cor:mu_threshold}.

\begin{corollary}[Rate of the smooth refit]
\label{cor:smooth_refit_rate}
Under the conditions of Theorem~\ref{thm:T2_select}, on the event $\widehat S=S$, let $D_S^{\rm all}=[\bar{\mathbf X},\mathbf Z_S]$ and suppose the eigenvalues of $(D_S^{\rm all})^\top D_S^{\rm all}/n$ are bounded away from zero and infinity. Put $d_F=p+1+s_vq_e$ and assume $d_F=o(n)$.  For the selected roughness parameter, define
\[
\delta_{F,n}=
\max_{k\in S}2\hat\lambda_2
\|\Omega\boldsymbol\theta_{0k}^\ast\|_2.
\]
Then the smooth refit in \eqref{eq:smooth_refit} satisfies
\begin{align}
\|\hat\gamma^{\,F}-\gamma_0\|_2^2+
\|\hat{\boldsymbol\theta}^{\,F}_S-
\boldsymbol\theta_{0S}^\ast\|_2^2
&=O_p\!\left(
\frac{d_F}{n}+s_vq_e^{-2r}+s_v\delta_{F,n}^2
\right),\label{eq:smooth_refit_parameter_rate}\\
\sum_{k\in S}
\|\hat g_k^{\,F}-g_{0k}\|_{L_2}^2
&=O_p\!\left(
\frac{d_F}{n}+s_vq_e^{-2r}+s_v\delta_{F,n}^2
\right).\label{eq:smooth_refit_function_rate}
\end{align}
The same order holds for the in-sample squared prediction error of the smooth refit. The refit removes the first-order group-penalty shrinkage, while its remaining deterministic shrinkage is the explicitly displayed roughness bias. The displayed order is an error bound; consistency additionally requires
\[
\frac{d_F}{n}+s_vq_e^{-2r}+s_v\delta_{F,n}^2\longrightarrow0.
\]
The present theory does not prove that EBIC automatically selects a roughness parameter satisfying this condition.
\end{corollary}

Corollary~\ref{cor:smooth_refit_rate} concerns the estimator used to construct coefficient curves and predictions.  Unlike the selection estimator, it no longer contains group-penalty bias, but it can retain roughness bias when $\hat\lambda_2$ is too large.

\begin{corollary}[Classification of constant and zero effects]
\label{cor:mu_threshold}
Under the conditions of Theorem~\ref{thm:T2_select} and Assumption~\ref{ass:A6}, choose a deterministic threshold satisfying
\[
r_{\mu,n}=o(\tau_N),\qquad \tau_N\to0.
\]
If
\[
\min_{k\in\mathcal S_{\mathrm{const}}}|\mu_{0k}|
>2\tau_N,
\]
then the thresholded refit following \eqref{eq:structure_refit} satisfies
\[
\Pr\{\widehat{\mathcal S}_{\mathrm{const}}
=\mathcal S_{\mathrm{const}},\
\widehat{\mathcal S}_{\mathrm{zero}}
=\mathcal S_{\mathrm{zero}}\}\longrightarrow1.
\]
\end{corollary}

Corollary~\ref{cor:mu_threshold} supplies the second stage in classifying zero, constant, and varying effects. It is conditional on recovery of the varying set and requires an identifiable full-rank refit; it does not justify thresholding an arbitrary generalized-inverse representative under rank deficiency.

\begin{theorem}[One-step debiasing and cluster-robust oracle inference]
\label{thm:T3_oracle}
Assume that the conditions of Theorem~\ref{thm:T2_select} and Corollary~\ref{cor:mu_threshold} hold. Let $S_c=\mathcal S_{\mathrm{const}}$ have fixed cardinality, $\mathbf X_c=\mathbf X_{S_c}$, $W=[\mathbf1,\mathbf X_S,\mathbf Z_S]$, $d_n=\dim(W)=1+s_v(1+q_e)$, and $\tilde{\mathbf X}_c=M_W\mathbf X_c$. Partition $\tilde{\mathbf X}_c$ into subject blocks $\tilde{\mathbf X}_{c,i}$ and let $\Sigma_i=\operatorname{Var}(\boldsymbol\varepsilon_i\mid\mathcal D)$. Suppose, in probability (or deterministically for a fixed design),
\begin{align}
A_{c,n}&=\tilde{\mathbf X}_c^\top\tilde{\mathbf X}_c/n
\longrightarrow\Sigma_c>0,\label{eq:t3_A}\\
\Gamma_{c,n}&=\frac1n\sum_{i=1}^N
\tilde{\mathbf X}_{c,i}^\top\Sigma_i\tilde{\mathbf X}_{c,i}
\longrightarrow\Gamma_c>0,\label{eq:t3_Gamma}
\end{align}
and $d_n/n\to0$.  Assume uniformly bounded conditional fourth moments for the subject errors, $\max_i\|\tilde{\mathbf X}_{c,i}\|_{\rm op}=O(1)$, and the cluster Lindeberg condition: for every $\epsilon>0$, 
\begin{equation}\label{eq:cluster_lindeberg}
\sum_{i=1}^N
\mathbb{E}\!\left[
\|\Gamma_{c,n}^{-1/2}\tilde{\mathbf X}_{c,i}^\top
\boldsymbol\varepsilon_i/\sqrt n\|_2^2
\mathbf1\!\left\{
\|\Gamma_{c,n}^{-1/2}\tilde{\mathbf X}_{c,i}^\top
\boldsymbol\varepsilon_i/\sqrt n\|_2>\epsilon
\right\}
\middle|\mathcal D\right]\to0.
\end{equation}
Define
\[
\hat{\boldsymbol\mu}^{\,or}_c
=A_{c,n}^{-1}\tilde{\mathbf X}_c^\top M_W\mathbf y/n.
\]
Assume
\begin{equation}\label{eq:t3_bias}
\sqrt n\left\|
A_{c,n}^{-1}\tilde{\mathbf X}_c^\top M_W\mathbf r/n
\right\|_2=o_p(1).
\end{equation}
A sufficient, directly interpretable condition for
\eqref{eq:t3_bias} is
\[
\sqrt{n s_v}\,q_e^{-r}\longrightarrow0,
\]
together with the bounded eigenvalues in \eqref{eq:t3_A} and a bounded operator norm for $\tilde{\mathbf X}_c/\sqrt n$. Root-$n$ inference requires this undersmoothing condition, which is stronger than the approximation condition needed for estimation or prediction. For the unconditional distributional statement, define $\hat{\boldsymbol\mu}^{\,DB}_c$ arbitrarily, say as the zero vector, on the complement of the correct structural recovery event. Then, on the event of correct structural recovery,
\begin{equation}\label{eq:db_oracle_equivalence}
\hat{\boldsymbol\mu}^{\,DB}_c
=\hat{\boldsymbol\mu}^{\,PS}_c
=\hat{\boldsymbol\mu}^{\,or}_c.
\end{equation}
Consequently, for every fixed vector $a$,
\[
\sqrt n\,a^\top
(\hat{\boldsymbol\mu}^{\,DB}_c-\boldsymbol\mu_{0,c})
\ \xrightarrow{d}\
N\!\left(0,a^\top\Sigma_c^{-1}\Gamma_c
\Sigma_c^{-1}a\right).
\]
If $\hat{\boldsymbol\varepsilon}_i$ are residuals from the correctly selected oracle refit, then
\[
\hat\Gamma_c=\frac1n\sum_{i=1}^N
(\tilde{\mathbf X}_{c,i}^\top\hat{\boldsymbol\varepsilon}_i)
(\tilde{\mathbf X}_{c,i}^\top\hat{\boldsymbol\varepsilon}_i)^\top
\]
obeys $\hat\Gamma_c-\Gamma_{c,n}=o_p(1)$. The matrix $A_{c,n}^{-1}\hat\Gamma_cA_{c,n}^{-1}/n$ consistently estimates the covariance of the one-step estimator. As with the coefficient estimator, the covariance estimator may be assigned any fixed value off the correct-recovery event without changing this conclusion.
\end{theorem}

Theorem~\ref{thm:T3_oracle} provides oracle inference for fixed-dimensional constant-effect contrasts after correct structural recovery. Equation~\eqref{eq:db_oracle_equivalence} is the bias-correction statement: orthogonalization removes dependence on the penalized nuisance starting value exactly, not merely asymptotically. The remaining first-order bias is the spline approximation term in \eqref{eq:t3_bias}. A spline dimension selected only for prediction need not make this term negligible at the root-$n$ scale, so inferential use requires undersmoothing or an additional approximation-bias correction. Different smoothing rates for estimation and inference also arise in longitudinal and high-dimensional VCM inference \citep{hu2021robust,dai2021vcqr}. The conclusion does not cover misspecified selection events or simultaneous confidence bands for the varying coefficient functions.

\section{Numerical Study}\label{sec:5}

\subsection{Data Design}

We evaluate TV-Select under the longitudinal varying-coefficient model
\[
y_{ij}
=
\beta_0+\sum_{k=1}^{p}x_{ijk}\{\mu_{0k}+g_{0k}(t_{ij})\}
+\varepsilon_{ij},
\qquad t_{ij}\in[0,1].
\]
We consider two sample-size designs: $(N,p)=(100,50)$ and $(N,p)=(200,100)$. Both use $n_i=8$ observations per subject, $q=8$ raw cubic B-spline basis functions, $\sigma=1$, and $R=200$ Monte Carlo replications. After centering and normalization, the effective basis dimension is $q_e=7$. Prediction is evaluated on an independently generated test sample of $N_{\mathrm{test}}=500$ subjects. Tables contain results for both designs; all simulation figures use $(N,p)=(200,100)$ unless explicitly stated otherwise. Because the baseline covariates are repeated within subject, the rank of $\bar{\mathbf X}$ in Scenarios A--D and F is at most $N$. The implementation enforces $p<N$ in these scenarios, a finite-sample rank condition that is stronger than merely requiring $p<n_\bullet$.

We set $\beta_0=0$.  In each scenario, six covariates have time-varying effects, six have nonzero constant effects, and the remaining $p-12$ are inactive.  For each combination of scenario and $p$, a single support permutation is generated from the prespecified seed 2026 before any data or tuning parameters are generated; its first six positions define $S_{\mathrm{vary}}$ and its next six define $S_{\mathrm{const}}$. This support is held fixed across replications and is never supplied to a fitting or tuning routine. The six centered varying effects are proportional to
\[
\sin(2\pi t),\quad \cos(2\pi t),\quad
\sin(4\pi t),\quad \cos(4\pi t),\quad
16t^2(1-t)^2-\frac{8}{15},\quad
\sin(\pi t)-\frac{2}{\pi},
\]
and each shape is normalized to unit $L_2[0,1]$ norm before multiplication by its scenario-specific amplitude. The six nonzero constant coefficients are split evenly between $1$ and $-1$. Baseline covariates follow a mean-zero Gaussian AR(1) design with $\operatorname{Cov}(x_{ik},x_{i\ell})=\rho^{|k-\ell|}$ and are repeated over visits within subject.  In Scenario E,
\[
x_{ijk}=x_{ik}+\delta_{ijk},\qquad
\delta_{ijk}\stackrel{\mathrm{iid}}{\sim}N(0,\sigma_x^2),
\]
independently of the baseline covariates, times, and errors.

The scale convention is fixed by the data-generating mechanism rather than estimated from each replication.  Baseline Gaussian covariates have unit marginal variance and are not empirically restandardized. In Scenario E, $\operatorname{Var}(x_{ijk})=1+\sigma_x^2=1.25$; this increase is part of the visit-level perturbation and is not removed by rescaling. The interaction blocks $\mathbf Z_k$ are not normalized separately. Instead, all methods use the same centered, full-rank, $L_2$-normalized effective spline basis, the same covariate matrices, and the same roughness operator. The independent test data are generated under the same scenario and evaluated with the training-defined analytic basis; no test-specific scaling is performed.

For irregular schedules, the visit times are independent $\mathrm{Unif}(0,1)$ draws sorted within subject.  For regular schedules, $t_{ij}=(j-1)/(n_i-1)$.  Gaussian errors are independent with variance $\sigma^2$, except in Scenario C, where
\[
\boldsymbol\varepsilon_i\sim
N\{\mathbf0,\sigma^2R_\alpha\},\qquad
(R_\alpha)_{j\ell}=\alpha^{|j-\ell|}.
\]
Scenario D uses $\varepsilon_{ij}=\sigma T_{ij}/\sqrt3$ with $T_{ij}\stackrel{\mathrm{iid}}{\sim}t_3$, so $\operatorname{Var}(\varepsilon_{ij})=\sigma^2$. Table~\ref{tab:sim_design} summarizes the six scenarios.

\begin{table}[H]
\centering
\caption{Simulation scenarios.}
\label{tab:sim_design}
\small
\setlength{\tabcolsep}{3pt}
\renewcommand{\arraystretch}{1.12}
\begin{tabular}{cllll}
\toprule
Scenario & Observation times & Covariates & Errors & Signal amplitudes $(a_1,\ldots,a_6)$\\
\midrule
A & Irregular & Baseline, $\rho=0.3$ & Gaussian
  & $(1,1,1,1,1,1)$\\
B & Irregular & Baseline, $\rho=0.6$ & Gaussian
  & $(1,.90,.75,.60,.45,.30)$\\
C & Regular & Baseline, $\rho=0.3$ & AR(1), $\alpha=0.6$
  & $(1,.90,.75,.60,.45,.30)$\\
D & Irregular & Baseline, $\rho=0.3$ & Standardized $t_3$
  & $(1,.85,.70,.55,.40,.25)$\\
E & Irregular & Time-varying, $\sigma_x=0.5$ & Gaussian
  & $(1,.85,.70,.55,.40,.25)$\\
F & Irregular & Baseline, $\rho=0.3$ & Gaussian
  & $(.65,.55,.45,.35,.25,.18)$\\
\bottomrule
\end{tabular}
\end{table}

Scenario D deliberately violates both the conditional sub-Gaussian condition in Assumption~\ref{ass:A1} and the bounded-fourth-moment condition used for Theorem~\ref{thm:T3_oracle}. It is treated as a robustness stress test rather than a numerical verification of the stated theory. We additionally examine sensitivity to the spline dimension $q\in\{8,10,12\}$ in Scenarios A and F at $(N,p)=(200,100)$. 

\subsection{Competing Methods}

We compare TV-Select with three procedures that distinguish the roles of structural selection, roughness regularization, and marginal screening.
\begin{itemize}
\item \textbf{TV-Select (proposed).}  TV-Select combines a block penalty for identifying time-varying effects with a roughness penalty for controlling their curvature.  The penalized fit determines the varying set, after which a smooth refit yields coefficient estimates and predictions.

\item \textbf{Group-Lasso.} Group-Lasso retains the spline model and group penalty \citep{yuan2006model} used by TV-Select but sets the roughness penalty to zero, thereby isolating the contribution of smoothness regularization.

\item \textbf{Marginal-VC.} Marginal-VC ranks predictors by the marginal improvement from a constant to a smooth varying-coefficient fit, reflecting marginal spline comparisons used in VCM selection \citep{wang2008variable,tang2013variable}. It jointly refits the top-ranked blocks while retaining all $p$ constant-effect columns. EBIC selects both the screening size, from 1 to 10, and the roughness parameter from the TV-Select grid.

\item \textbf{VC-Ridge.} VC-Ridge removes the block penalty and represents all $p$ effects as time-varying, with P-spline roughness regularization \citep{eilers1996flexible} and a small ridge stabilization. Because it does not select a varying set, it enters only the estimation and prediction comparisons.
\end{itemize}

All methods use the same normalized effective cubic B-spline basis, the same training and test samples, and the same preprocessing.  The roughness grid is $\{10^{-3},10^{-2},10^{-1}\}$.  For TV-Select and Group-Lasso, the $\lambda_1$ path is generated from the data, beginning at the data-derived $\lambda_{\max}$ that sets every varying block to zero and decreasing to $0.03\lambda_{\max}$ on a logarithmic grid; the true support is never used for tuning. Method-specific tuning parameters are selected by
\[
\operatorname{EBIC}
=
\log\!\left(\frac{\operatorname{RSS}}{n_\bullet}\right)
+\frac{\log(n_\bullet)}{n_\bullet}\widehat{\operatorname{df}}
+\frac{2\gamma}{n_\bullet}
 \log\binom{p}{|\widehat{S}_{\mathrm{vary}}|},
\qquad \gamma=0.5,
\]
where $n_\bullet=\sum_i n_i$, $\widehat{\operatorname{df}}$ is the effective model dimension, and the final term accounts for model-space multiplicity at the selected varying-set size.  Specifically, for the selection-capable methods, let $\tilde{\mathbf Z}_A=M_X\mathbf Z_A$,  $G_A=\tilde{\mathbf Z}_A^\top\tilde{\mathbf Z}_A/n_\bullet$, and $P_A=2\lambda_2(I_{|A|}\otimes\Omega)$.  The implementation uses
\[
\widehat{\operatorname{df}}(A,\lambda_2)
=\operatorname{rank}(\bar{\mathbf X})
+\operatorname{tr}\{(G_A+P_A)^{-1}G_A\}.
\]
For an unpenalized full-rank refit, this expression reduces to $p+1+q_e|A|$. For the nonselective VC-Ridge fit, storing and inverting the full $pq_e$-dimensional smoother matrix at every Monte Carlo tuning step would be prohibitive; its implementation uses the declared blockwise approximation
\[
\operatorname{rank}(\bar{\mathbf X})
+\sum_{k=1}^p
\operatorname{tr}\{(G_{kk}+2\lambda_2\Omega_R)^{-1}G_{kk}\},
\]
where $G_{kk}$ is the residualized block Gram matrix and $\Omega_R$ is its ridge-augmented roughness matrix. VC-Ridge has no model-space term because it always includes every varying block.

\subsection{Evaluation Metrics}

Let $S_{\mathrm{vary}}$, $S_{\mathrm{const}}$, and $S_{\mathrm{zero}}$ denote the true structural partition, and let $\widehat S_{\mathrm{vary}}$ be the varying set selected by a method.  We use the following measures.
\begin{itemize}
\item \textbf{Varying-set recovery.}  Sensitivity and false discovery of
time variation are summarized by
\[
\operatorname{TPR}_{\mathrm{vary}}
=\frac{|\widehat S_{\mathrm{vary}}\cap S_{\mathrm{vary}}|}
{|S_{\mathrm{vary}}|},
\qquad
\operatorname{FPR}_{\mathrm{vary}}
=\frac{|\widehat S_{\mathrm{vary}}\cap S_{\mathrm{vary}}^c|}
{|S_{\mathrm{vary}}^c|}.
\]
Exact recovery is evaluated by $\operatorname{Exact}_{\mathrm{vary}}=\mathbf{1}\{\widehat S_{\mathrm{vary}}=S_{\mathrm{vary}}\}$. Larger TPR and exact-recovery probability are preferable, whereas a smaller FPR is preferable. Selection reproducibility is measured by the average pairwise Jaccard index across the $R$ replications,
\[
\operatorname{Stab}
=\frac{2}{R(R-1)}\sum_{1\le r<s\le R}
\frac{|\widehat S_{\mathrm{vary}}^{(r)}\cap
        \widehat S_{\mathrm{vary}}^{(s)}|}
     {|\widehat S_{\mathrm{vary}}^{(r)}\cup
        \widehat S_{\mathrm{vary}}^{(s)}|},
\]
where a value closer to one indicates more reproducible selection.

\item \textbf{Three-class recovery.}  For every method capable of selecting a varying set, the non-varying coordinates are classified by the common, truth-independent threshold
\[
\tau_n=\{\log(n_\bullet)\}^{1/4}
\sqrt{\frac{\log p}{n_\bullet}}.
\]
Writing $c_k\in\{\mathrm{zero},\mathrm{const},\mathrm{vary}\}$ and $\hat c_k$ for the true and estimated labels, respectively, we record
\[
\operatorname{ClassAcc}
=\frac1p\sum_{k=1}^p\mathbf1(\hat c_k=c_k),\qquad
\operatorname{Exact}_3
=\mathbf1\{(\widehat S_{\mathrm{zero}},
\widehat S_{\mathrm{const}},\widehat S_{\mathrm{vary}})
=(S_{\mathrm{zero}},S_{\mathrm{const}},S_{\mathrm{vary}})\}.
\]
The same threshold rule is used for all methods and is not calibrated to the true coefficient magnitudes. TV-Select uses the identifiable classification refit in \eqref{eq:structure_refit}; if that refit is rank deficient, the replication is flagged rather than classified.

\item \textbf{Active-effect estimation.}  For the nonzero constant effects, we compute
\[
\operatorname{MSE}_{\mu,\mathrm{act}}
=\frac{1}{|S_{\mathrm{const}}|}
 \sum_{k\in S_{\mathrm{const}}}(\widehat\mu_k-\mu_{0k})^2.
\]
Let $S_{\mathrm{act}}=S_{\mathrm{vary}}\cup S_{\mathrm{const}}$ and $\widehat\beta_k(t)=\widehat\mu_k+\boldsymbol C(t)^\top\widehat{\boldsymbol\theta}_k$.  On an equally spaced grid $\{t_g\}_{g=1}^{G}$ with $G=200$, active-function error is
\[
\operatorname{ISE}_{\mathrm{act}}
=\frac{1}{|S_{\mathrm{act}}|}
 \sum_{k\in S_{\mathrm{act}}}\frac{1}{G}
 \sum_{g=1}^{G}
 \{\widehat\beta_k(t_g)-\beta_{0k}(t_g)\}^2.
\]
Both measures are nonnegative and smaller values indicate more accurate estimation. Table~\ref{tab:sim_estimation} multiplies them by $10^3$ for readability.

\item \textbf{Curvature estimation.}  Curvature recovery is evaluated by the roughness error
\[
\operatorname{RE}
=\frac{1}{|S_{\mathrm{vary}}|}
\sum_{k\in S_{\mathrm{vary}}}
\int_0^1\{\widehat g_k''(t)-g_{0k}''(t)\}^2\,dt.
\]
Smaller values indicate more accurate recovery of the second derivatives.

\item \textbf{Prediction.}  An independent test sample with $N_{\mathrm{test}}=500$ subjects is generated from the same scenario.  Prediction accuracy is measured by
\[
\operatorname{MSPE}
=\frac{1}{n_{\mathrm{test}}}
 \sum_{\ell=1}^{n_{\mathrm{test}}}
 (y^{\mathrm{test}}_\ell-\widehat y^{\mathrm{test}}_\ell)^2,
\]
where smaller values indicate better out-of-sample prediction.
\end{itemize}

Except for stability, which is computed directly from all $R$ selected sets, table entries are Monte Carlo means with standard errors in parentheses, estimated as the sample standard deviation divided by $\sqrt{R}$.

\subsection{Main Results}

Table~\ref{tab:sim_selection} gives the numerical results for both sample-size designs, and Figure~\ref{fig:sim_exact} compares exact-recovery probabilities for $(N,p)=(200,100)$. In the figure, the bars are Monte Carlo proportions and the error bars are 95\% Wilson binomial intervals, which remain nondegenerate when all replications succeed. VC-Ridge is not included because, by construction, it retains all $p$ varying blocks and therefore does not perform structural selection.

\begin{sidewaystable}[p]
\centering
\caption{Structural recovery under two sample-size designs.}
\label{tab:sim_selection}
\renewcommand{\arraystretch}{1.08}
\resizebox{\textwidth}{!}{
\begin{tabular}{clccccc@{\hspace{6pt}}ccccc}
\toprule
& & \multicolumn{5}{c}{$N=100,\ p=50$} & \multicolumn{5}{c}{$N=200,\ p=100$}\\
\cmidrule(lr){3-7} \cmidrule(lr){8-12}
Scenario & Method & TPR & FPR & Exact & ClassAcc & Stability & TPR & FPR & Exact & ClassAcc & Stability\\
\midrule
\multirow{3}{*}{A} & TV-Select & 1.000 (0.000) & 0.0001 (0.0001) & 0.995 (0.005) & 0.9582 (0.0024) & 0.999 & 1.000 (0.000) & 0.0000 (0.0000) & 1.000 (0.000) & 0.9781 (0.0014) & 1.000
\\
 & Group-Lasso & 1.000 (0.000) & 0.0233 (0.0016) & 0.380 (0.034) & 0.9370 (0.0029) & 0.774 & 1.000 (0.000) & 0.0051 (0.0006) & 0.720 (0.032) & 0.9647 (0.0016) & 0.887
\\
 & Marginal-VC & 0.996 (0.002) & 0.0060 (0.0011) & 0.800 (0.028) & 0.9505 (0.0027) & 0.927 & 1.000 (0.000) & 0.0001 (0.0001) & 0.995 (0.005) & 0.9780 (0.0014) & 0.999
\\
\addlinespace[2pt]
\multirow{3}{*}{B} & TV-Select & 0.937 (0.006) & 0.0052 (0.0008) & 0.510 (0.035) & 0.8798 (0.0039) & 0.855 & 0.996 (0.002) & 0.0006 (0.0002) & 0.915 (0.020) & 0.9135 (0.0028) & 0.975
\\
 & Group-Lasso & 0.975 (0.004) & 0.0294 (0.0021) & 0.255 (0.031) & 0.8657 (0.0039) & 0.705 & 1.000 (0.000) & 0.0086 (0.0007) & 0.470 (0.035) & 0.8940 (0.0027) & 0.816
\\
 & Marginal-VC & 0.662 (0.009) & 0.0417 (0.0026) & 0.000 (0.000) & 0.7877 (0.0047) & 0.554 & 0.740 (0.007) & 0.0254 (0.0013) & 0.005 (0.005) & 0.8495 (0.0032) & 0.620
\\
\addlinespace[2pt]
\multirow{3}{*}{C} & TV-Select & 1.000 (0.000) & 0.0003 (0.0002) & 0.985 (0.009) & 0.8124 (0.0044) & 0.996 & 1.000 (0.000) & 0.0000 (0.0000) & 1.000 (0.000) & 0.8275 (0.0034) & 1.000
\\
 & Group-Lasso & 1.000 (0.000) & 0.0372 (0.0025) & 0.335 (0.033) & 0.7869 (0.0045) & 0.693 & 1.000 (0.000) & 0.0170 (0.0014) & 0.400 (0.035) & 0.8143 (0.0035) & 0.705
\\
 & Marginal-VC & 0.969 (0.005) & 0.0101 (0.0012) & 0.535 (0.035) & 0.8026 (0.0044) & 0.846 & 0.975 (0.004) & 0.0096 (0.0008) & 0.365 (0.034) & 0.8192 (0.0033) & 0.803
\\
\addlinespace[2pt]
\multirow{3}{*}{D} & TV-Select & 0.915 (0.008) & 0.0008 (0.0003) & 0.495 (0.035) & 0.9469 (0.0039) & 0.896 & 0.994 (0.002) & 0.0005 (0.0002) & 0.925 (0.019) & 0.9755 (0.0020) & 0.976
\\
 & Group-Lasso & 0.960 (0.007) & 0.0172 (0.0016) & 0.385 (0.034) & 0.9398 (0.0039) & 0.771 & 0.998 (0.001) & 0.0058 (0.0007) & 0.670 (0.033) & 0.9643 (0.0022) & 0.868
\\
 & Marginal-VC & 0.752 (0.009) & 0.0132 (0.0014) & 0.015 (0.009) & 0.9071 (0.0041) & 0.696 & 0.843 (0.007) & 0.0081 (0.0008) & 0.065 (0.017) & 0.9555 (0.0022) & 0.736
\\
\addlinespace[2pt]
\multirow{3}{*}{E} & TV-Select & 0.948 (0.005) & 0.0001 (0.0001) & 0.685 (0.033) & 0.9889 (0.0009) & 0.927 & 0.999 (0.001) & 0.0001 (0.0001) & 0.990 (0.007) & 0.9979 (0.0004) & 0.997
\\
 & Group-Lasso & 0.988 (0.003) & 0.0187 (0.0015) & 0.415 (0.035) & 0.9782 (0.0014) & 0.791 & 0.999 (0.001) & 0.0029 (0.0005) & 0.815 (0.028) & 0.9936 (0.0007) & 0.930
\\
 & Marginal-VC & 0.788 (0.008) & 0.0120 (0.0014) & 0.050 (0.015) & 0.9548 (0.0018) & 0.735 & 0.871 (0.007) & 0.0054 (0.0007) & 0.155 (0.026) & 0.9847 (0.0008) & 0.788
\\
\addlinespace[2pt]
\multirow{3}{*}{F} & TV-Select & 0.756 (0.007) & 0.0008 (0.0003) & 0.020 (0.010) & 0.9236 (0.0027) & 0.866 & 0.932 (0.007) & 0.0002 (0.0001) & 0.625 (0.034) & 0.9738 (0.0013) & 0.902
\\
 & Group-Lasso & 0.805 (0.007) & 0.0103 (0.0012) & 0.050 (0.015) & 0.9233 (0.0029) & 0.777 & 0.974 (0.004) & 0.0027 (0.0005) & 0.680 (0.033) & 0.9661 (0.0015) & 0.894
\\
 & Marginal-VC & 0.603 (0.010) & 0.0082 (0.0010) & 0.005 (0.005) & 0.8928 (0.0032) & 0.661 & 0.733 (0.009) & 0.0043 (0.0005) & 0.035 (0.013) & 0.9543 (0.0015) & 0.725
\\
\bottomrule
\end{tabular}}
\end{sidewaystable}

\begin{figure}[H]
\centering
\includegraphics[width=0.8\textwidth]{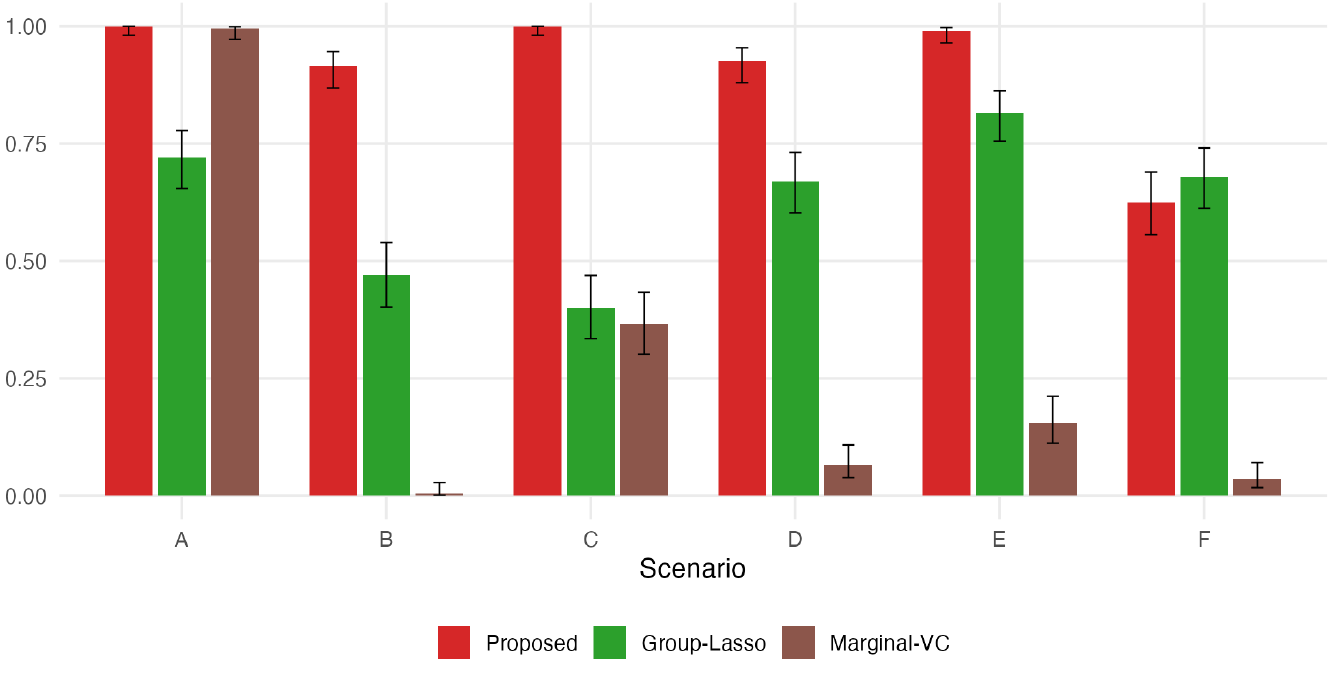}
\caption{Exact recovery of the time-varying set for
$(N,p)=(200,100)$.}
\label{fig:sim_exact}
\end{figure}

For $(N,p)=(200,100)$, TV-Select attains exact recovery with probability $1.000$ in Scenarios A and C, $0.990$ in Scenario E, and above $0.90$ in Scenarios B and D, while its FPR never exceeds $6.4\times10^{-4}$. Group-Lasso has nearly perfect TPR, but its higher FPR reduces exact recovery to $0.400$--$0.815$ in Scenarios B--E.  Scenario F is the only setting in which Group-Lasso has higher exact recovery ($0.680$ versus $0.625$); that gain is accompanied by a larger FPR ($0.0027$ versus $0.0002$) and slightly lower stability.

The smaller design gives the expected finite-sample deterioration without changing the qualitative comparison. For $(N,p)=(100,50)$, TV-Select has the highest exact-recovery probability in Scenarios A--E and the highest ClassAcc and stability in all six scenarios. Its FPR ranges from $0.0001$ to $0.0052$, compared with $0.0103$--$0.0372$ for Group-Lasso across Scenarios A--F. Scenario F exposes the detection boundary: exact recovery is only $0.020$ for TV-Select and $0.050$ for Group-Lasso because the weakest varying effects are frequently missed. Marginal-VC remains competitive in Scenario A but deteriorates under stronger dependence and weak signals. TV-Select retains most active blocks while more strongly suppressing spurious time variation. Because most variables are inactive in both designs, ClassAcc is interpreted together with exact recovery rather than in isolation.

Table~\ref{tab:sim_estimation} summarizes active-effect estimation and out-of-sample prediction for both sample-size designs, and Figure~\ref{fig:sim_errors} compares the methods for $(N,p)=(200,100)$. Bar heights are replication averages, and the error bars are approximate 95\% Monte Carlo confidence intervals. The vertical axes use $\log(1+x)$ only to improve visual separation. VC-Ridge is omitted from these panels because its much larger errors would compress the other bars; its complete results remain in Table~\ref{tab:sim_estimation}.

For $(N,p)=(200,100)$, TV-Select has the smallest active-effect ISE and MSPE in Scenarios B--F. In Scenario A, Marginal-VC has active-effect ISE $0.00663$ versus $0.00683$ and MSPE $1.194$ versus $1.197$. Compared with Group-Lasso, TV-Select reduces active-effect ISE by approximately 62--71\% and MSPE by 9--14\% across the six scenarios. At $(N,p)=(100,50)$, TV-Select has the smallest active-effect ISE and MSPE in every scenario and the smallest active mean-effect MSE except in Scenario C, where the values are nearly equal. The lower Marginal-VC error in Scenario A is not maintained under dependence, non-Gaussian errors, time-varying covariates, or weak signals. VC-Ridge has larger errors in both designs, illustrating the variance cost of treating every effect as time varying.

The RE comparison differs from those based on ISE and MSPE. Group-Lasso has the smallest RE in most scenarios under both designs, whereas TV-Select generally improves ISE and MSPE more strongly. This discrepancy reflects a tradeoff between bias and variance: roughness regularization stabilizes coefficient functions and predictions but can shrink second derivatives toward zero. RE is used as a curvature-recovery diagnostic rather than as a substitute for function-level error or predictive performance.

\begin{table}[!tbp]
\centering
\caption{Active-effect estimation, curvature recovery, and prediction.}
\label{tab:sim_estimation}
\footnotesize
\setlength{\tabcolsep}{3.0pt}
\renewcommand{\arraystretch}{1.08}
\resizebox{\textwidth}{!}{
\begin{tabular}{clrrrr@{\hspace{9pt}}rrrr}
\toprule
& & \multicolumn{4}{c}{$N=100,\ p=50$} & \multicolumn{4}{c}{$N=200,\ p=100$}\\
\cmidrule(lr){3-6} \cmidrule(lr){7-10}
Scenario & Method & $10^3\mathrm{MSE}_{\mu,\mathrm{act}}$ & $10^3\mathrm{ISE}_{\mathrm{act}}$ & $10^{-3}\mathrm{RE}$ & MSPE & $10^3\mathrm{MSE}_{\mu,\mathrm{act}}$ & $10^3\mathrm{ISE}_{\mathrm{act}}$ & $10^{-3}\mathrm{RE}$ & MSPE\\
\midrule
\multirow{4}{*}{A} & TV-Select & 3.429 (0.167) & 11.978 (0.166) & 4.218 (0.078) & 1.243 (0.004) & 1.536 (0.062) & 6.827 (0.072) & 3.400 (0.046) & 1.197 (0.003)
\\
 & Group-Lasso & 4.122 (0.205) & 33.720 (0.701) & 3.514 (0.067) & 1.521 (0.010) & 1.737 (0.069) & 19.176 (0.307) & 2.814 (0.039) & 1.357 (0.005)
\\
 & Marginal-VC & 3.536 (0.186) & 14.509 (1.055) & 5.018 (0.106) & 1.278 (0.014) & 1.535 (0.061) & 6.633 (0.070) & 3.855 (0.053) & 1.194 (0.003)
\\
 & VC-Ridge & 5.792 (0.302) & 46.308 (1.551) & 5.816 (0.262) & 2.747 (0.093) & 2.610 (0.099) & 25.509 (0.299) & 4.449 (0.078) & 2.612 (0.012)
\\
\hline\addlinespace[2pt]
\multirow{4}{*}{B} & TV-Select & 6.078 (0.248) & 16.138 (0.450) & 2.556 (0.057) & 1.257 (0.006) & 2.784 (0.119) & 6.747 (0.133) & 1.951 (0.032) & 1.179 (0.003)
\\
 & Group-Lasso & 7.358 (0.317) & 37.696 (0.859) & 1.671 (0.030) & 1.515 (0.011) & 3.043 (0.132) & 17.555 (0.261) & 1.330 (0.021) & 1.317 (0.004)
\\
 & Marginal-VC & 7.663 (0.333) & 47.478 (1.303) & 3.511 (0.095) & 1.641 (0.017) & 3.194 (0.129) & 27.319 (0.841) & 2.312 (0.046) & 1.438 (0.011)
\\
 & VC-Ridge & 10.221 (0.529) & 89.668 (12.503) & 8.936 (1.563) & 3.067 (0.232) & 4.389 (0.194) & 28.833 (1.244) & 2.734 (0.118) & 2.290 (0.044)
\\
\hline\addlinespace[2pt]
\multirow{4}{*}{C} & TV-Select & 9.644 (0.417) & 15.734 (0.350) & 2.567 (0.031) & 1.447 (0.009) & 4.925 (0.211) & 9.128 (0.160) & 2.440 (0.019) & 1.420 (0.007)
\\
 & Group-Lasso & 9.641 (0.417) & 49.910 (0.951) & 2.537 (0.042) & 1.798 (0.013) & 4.925 (0.211) & 30.935 (0.500) & 1.956 (0.023) & 1.649 (0.009)
\\
 & Marginal-VC & 9.640 (0.417) & 17.895 (0.479) & 2.681 (0.038) & 1.476 (0.010) & 4.925 (0.211) & 10.414 (0.272) & 2.496 (0.023) & 1.441 (0.008)
\\
 & VC-Ridge & 10.605 (0.471) & 25.468 (0.454) & 2.299 (0.029) & 1.958 (0.011) & 5.395 (0.242) & 13.821 (0.207) & 2.142 (0.018) & 1.933 (0.009)
\\
\hline\addlinespace[2pt]
\multirow{4}{*}{D} & TV-Select & 3.551 (0.225) & 13.230 (1.336) & 2.344 (0.064) & 1.238 (0.025) & 1.586 (0.076) & 5.118 (0.127) & 1.741 (0.032) & 1.175 (0.019)
\\
 & Group-Lasso & 3.977 (0.243) & 30.701 (1.461) & 1.557 (0.036) & 1.462 (0.026) & 1.731 (0.085) & 15.106 (0.404) & 1.155 (0.019) & 1.305 (0.019)
\\
 & Marginal-VC & 3.896 (0.252) & 25.345 (1.498) & 2.737 (0.073) & 1.399 (0.026) & 1.671 (0.079) & 11.395 (0.407) & 1.848 (0.038) & 1.260 (0.020)
\\
 & VC-Ridge & 7.176 (0.684) & 53.250 (7.259) & 4.911 (0.610) & 3.165 (0.322) & 2.786 (0.258) & 22.602 (2.478) & 2.835 (0.291) & 2.805 (0.271)
\\
\hline\addlinespace[2pt]
\multirow{4}{*}{E} & TV-Select & 1.673 (0.076) & 8.445 (0.184) & 2.194 (0.048) & 1.198 (0.004) & 0.760 (0.030) & 3.681 (0.053) & 1.659 (0.028) & 1.135 (0.002)
\\
 & Group-Lasso & 1.927 (0.088) & 21.390 (0.517) & 1.442 (0.028) & 1.403 (0.009) & 0.831 (0.033) & 12.550 (0.196) & 1.123 (0.017) & 1.278 (0.004)
\\
 & Marginal-VC & 1.849 (0.082) & 17.993 (0.642) & 2.434 (0.059) & 1.354 (0.010) & 0.823 (0.033) & 8.598 (0.355) & 1.767 (0.034) & 1.215 (0.006)
\\
 & VC-Ridge & 3.183 (0.141) & 27.805 (0.410) & 3.740 (0.086) & 2.453 (0.014) & 1.407 (0.056) & 14.681 (0.184) & 2.652 (0.051) & 2.456 (0.011)
\\
\hline\addlinespace[2pt]
\multirow{4}{*}{F} & TV-Select & 3.231 (0.153) & 14.327 (0.299) & 1.531 (0.047) & 1.271 (0.005) & 1.686 (0.087) & 5.574 (0.144) & 0.963 (0.021) & 1.178 (0.003)
\\
 & Group-Lasso & 3.655 (0.187) & 31.139 (0.635) & 0.880 (0.018) & 1.459 (0.009) & 1.842 (0.094) & 14.594 (0.211) & 0.591 (0.009) & 1.297 (0.004)
\\
 & Marginal-VC & 3.519 (0.156) & 23.289 (0.641) & 1.941 (0.066) & 1.388 (0.009) & 1.788 (0.087) & 11.211 (0.311) & 1.070 (0.027) & 1.255 (0.005)
\\
 & VC-Ridge & 5.211 (0.231) & 35.632 (0.496) & 2.815 (0.067) & 2.474 (0.015) & 2.824 (0.136) & 18.155 (0.227) & 1.653 (0.038) & 2.466 (0.010)
\\
\bottomrule
\end{tabular}}
\end{table}

\begin{figure}[H]
\centering
\begin{subfigure}{0.82\textwidth}
\centering
\includegraphics[width=\linewidth]{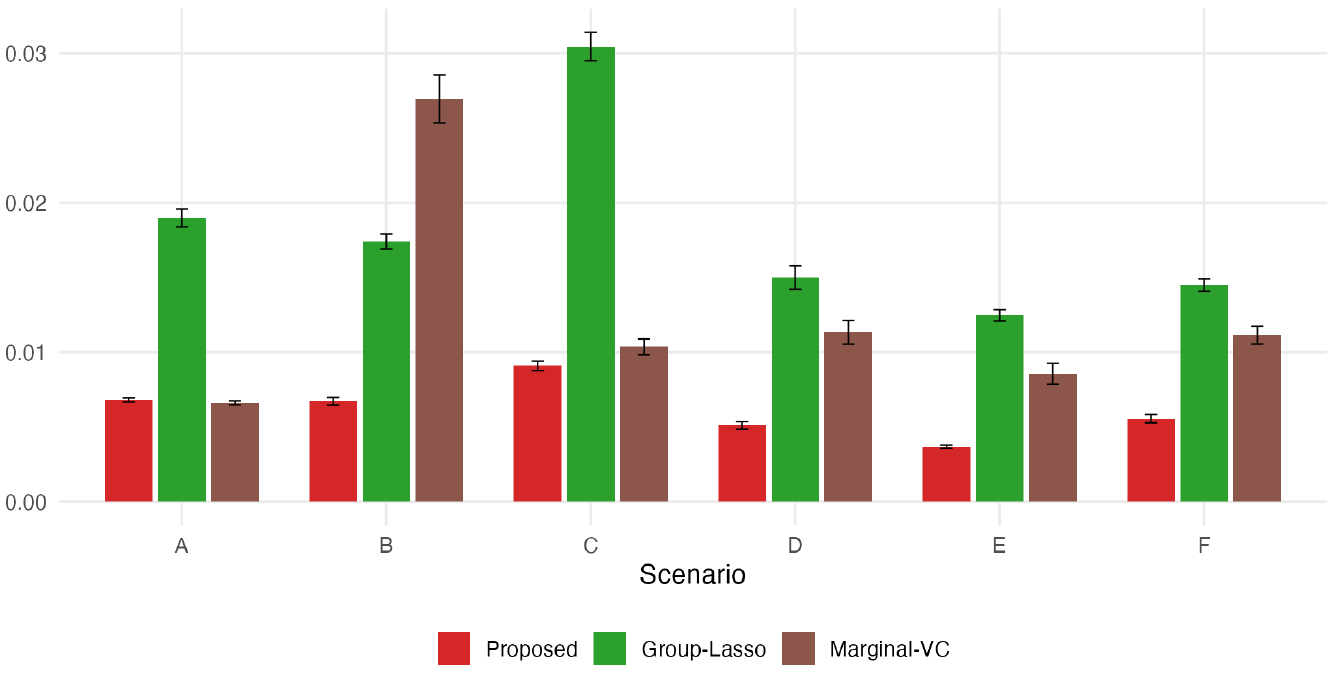}
\caption{Active-effect integrated squared error.}
\label{fig:sim_errors_ise}
\end{subfigure}

\vspace{0.5em}

\begin{subfigure}{0.82\textwidth}
\centering
\includegraphics[width=\linewidth]{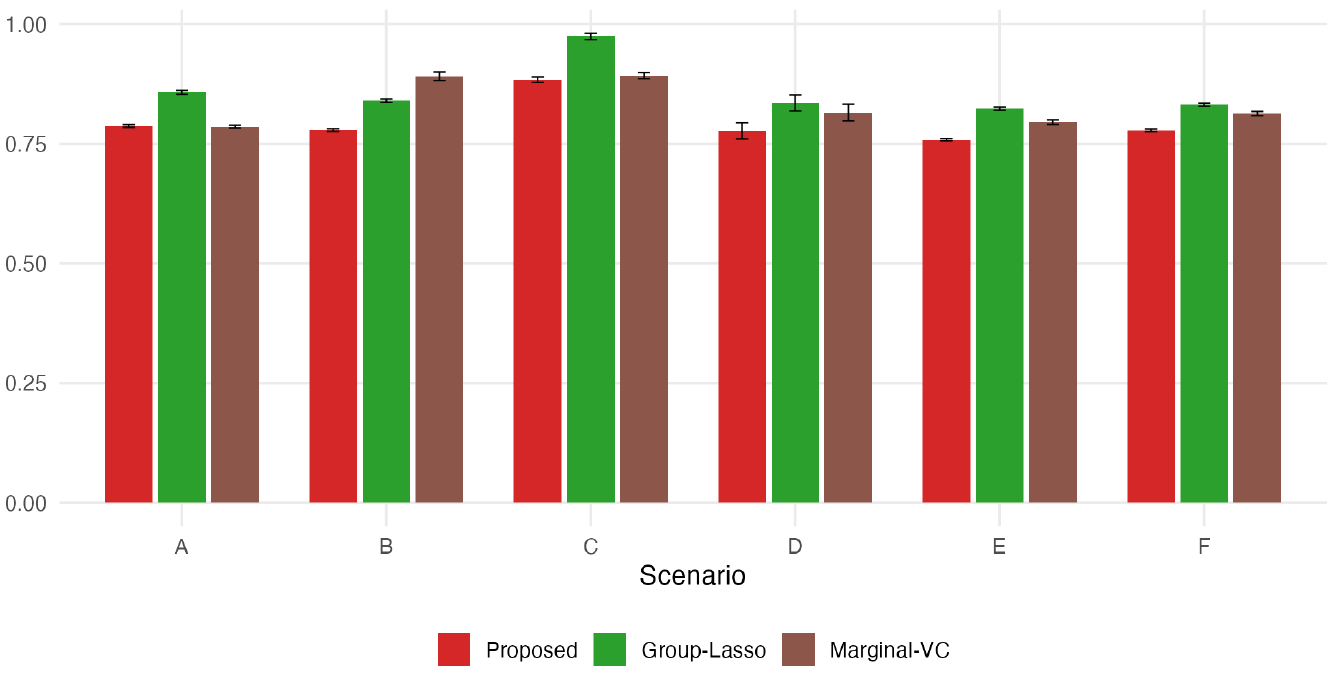}
\caption{Out-of-sample mean squared prediction error.}
\label{fig:sim_errors_mspe}
\end{subfigure}
\caption{Estimation and prediction performance for $(N,p)=(200,100)$.}
\label{fig:sim_errors}
\end{figure}

Figure~\ref{fig:sim_curves} examines function recovery using estimates averaged across all replications rather than a selected best-fit replication. The black curves are the true functions, the red curves are the TV-Select averages, and the shaded regions are the pointwise 10th--90th percentile bands over the $R=200$ replications. In Scenario A, TV-Select closely recovers the low-frequency and unimodal effects; the most visible attenuation occurs for the higher-frequency fourth curve. In Scenario F, the first five shapes remain well recovered, whereas the weakest sixth effect is more strongly shrunk and has a wider band. The latter pattern is consistent with the lower TPR in Scenario F and identifies a finite-sample detection boundary. The averaged estimates retain the principal shapes without the oscillation and variance inflation reflected in the competing methods' ISE values.

\begin{figure}[H]
\centering
\begin{subfigure}{0.92\textwidth}
\centering
\includegraphics[width=\linewidth]{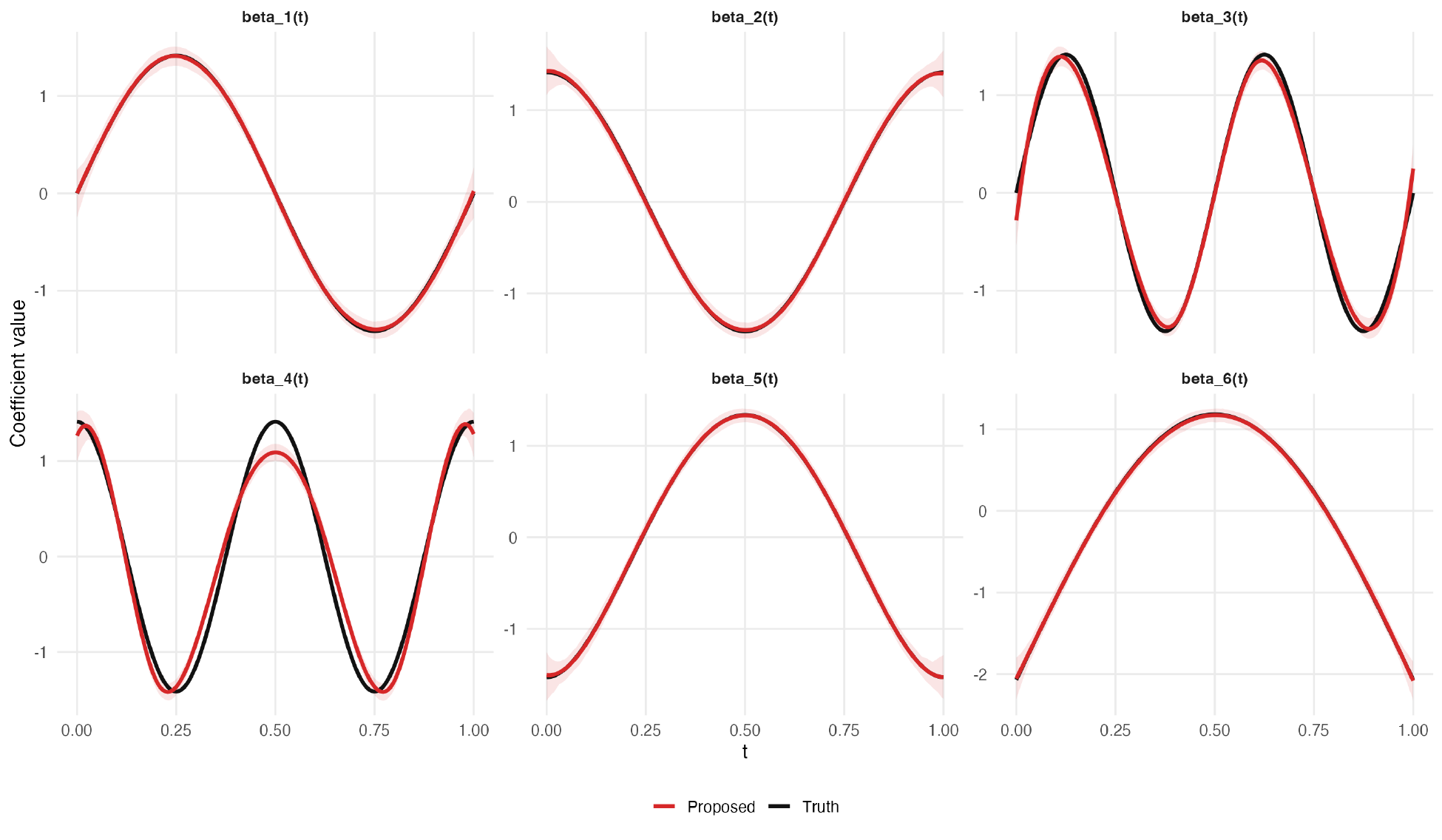}
\caption{Scenario A: baseline signal strength.}
\end{subfigure}

\vspace{0.5em}
\begin{subfigure}{0.92\textwidth}
\centering
\includegraphics[width=\linewidth]{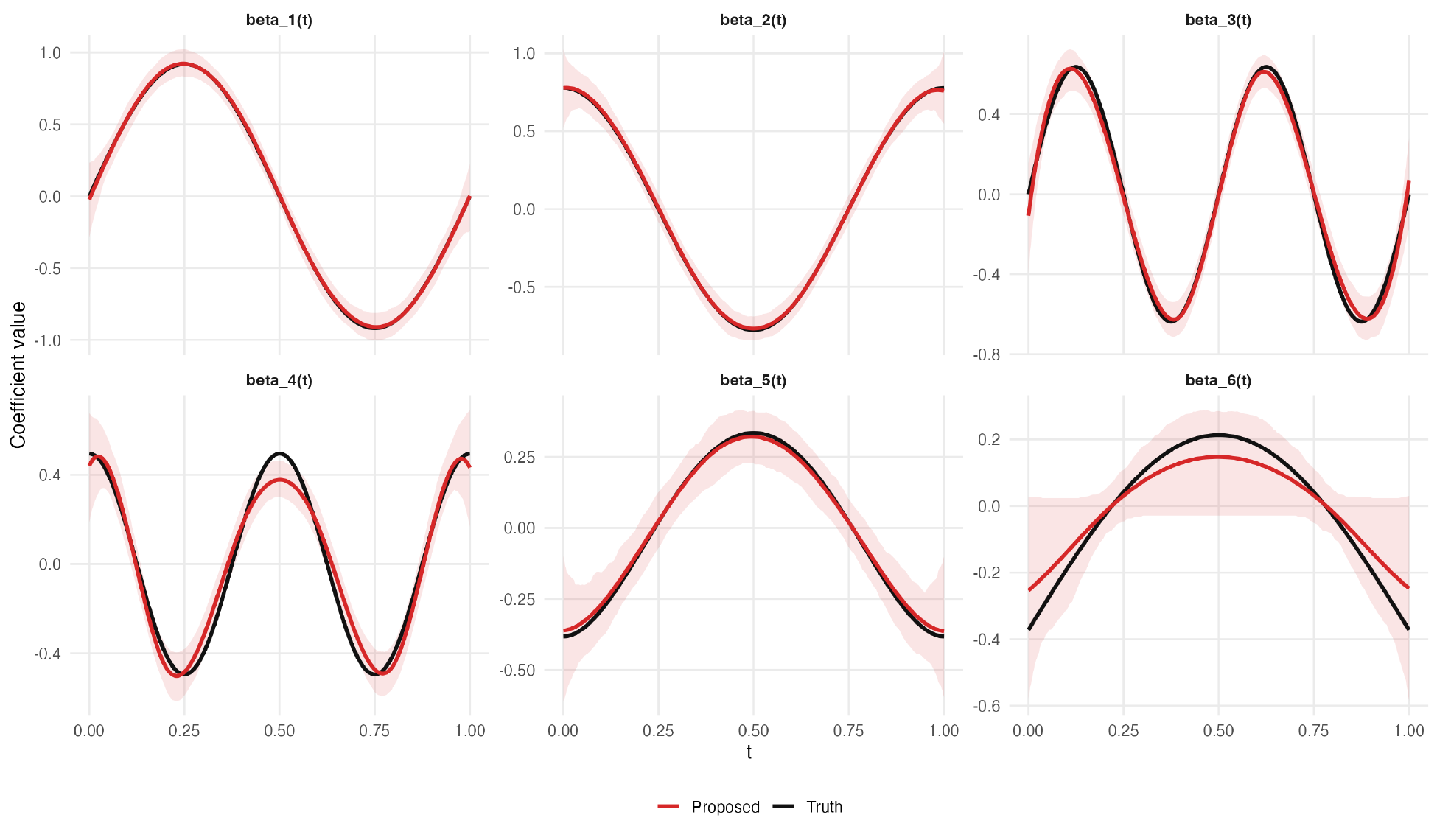}
\caption{Scenario F: decreasing and weak signal strengths.}
\end{subfigure}
\caption{TV-Select estimates averaged over 200 replications for $(N,p)=(200,100)$.}
\label{fig:sim_curves}
\end{figure}

Table~\ref{tab:sim_q_sensitivity} and Figure~\ref{fig:sim_q_sensitivity} assess sensitivity to the raw spline dimension $q$. The $q=8$ row is taken directly from the main $R=200$ experiment and coincides with Tables~\ref{tab:sim_selection} and \ref{tab:sim_estimation}; the $q=10$ and $q=12$ rows use the corresponding $R=100$ additional experiments. The figure displays approximate 95\% Monte Carlo confidence intervals using the replication count for each value of $q$. Active-effect ISE is multiplied by $10^3$. In Scenario A, support recovery is perfect for all three dimensions, while $q=10$ yields the smallest active-effect ISE and MSPE. In the weak-signal Scenario F, increasing $q$ reduces TPR and exact recovery because each varying block contains more coefficients to estimate from the same sample. FPR remains essentially zero. Greater flexibility improves estimation for strong signals but increases detection variability near the weak-signal boundary.

\begin{table}[H]
\centering
\caption{Sensitivity of TV-Select to the raw spline dimension.}
\label{tab:sim_q_sensitivity}
\footnotesize
\setlength{\tabcolsep}{4.2pt}
\begin{tabular}{ccccccc}
\toprule
Scenario & $q$ & TPR & FPR & Exact recovery &
$10^3\mathrm{ISE}_{\mathrm{act}}$ & MSPE\\
\midrule
\multirow{3}{*}{A}
 & 8  & 1.000 (0.000) & 0.0000 (0.0000) & 1.000 (0.000)
 & 6.827 (0.072) & 1.197 (0.003)\\
 & 10 & 1.000 (0.000) & 0.0000 (0.0000) & 1.000 (0.000)
 & 4.900 (0.107) & 1.169 (0.004)\\
 & 12 & 1.000 (0.000) & 0.0000 (0.0000) & 1.000 (0.000)
 & 5.519 (0.112) & 1.177 (0.004)\\
\addlinespace[2pt]
\multirow{3}{*}{F}
 & 8  & 0.933 (0.007) & 0.0002 (0.0001) & 0.625 (0.034)
 & 5.574 (0.144) & 1.178 (0.003)\\
 & 10 & 0.892 (0.011) & 0.0001 (0.0001) & 0.450 (0.050)
 & 6.597 (0.243) & 1.189 (0.005)\\
 & 12 & 0.843 (0.012) & 0.0000 (0.0000) & 0.290 (0.046)
 & 8.259 (0.280) & 1.210 (0.005)\\
\bottomrule
\end{tabular}
\end{table}

\begin{figure}[H]
\centering
\includegraphics[width=0.93\textwidth]{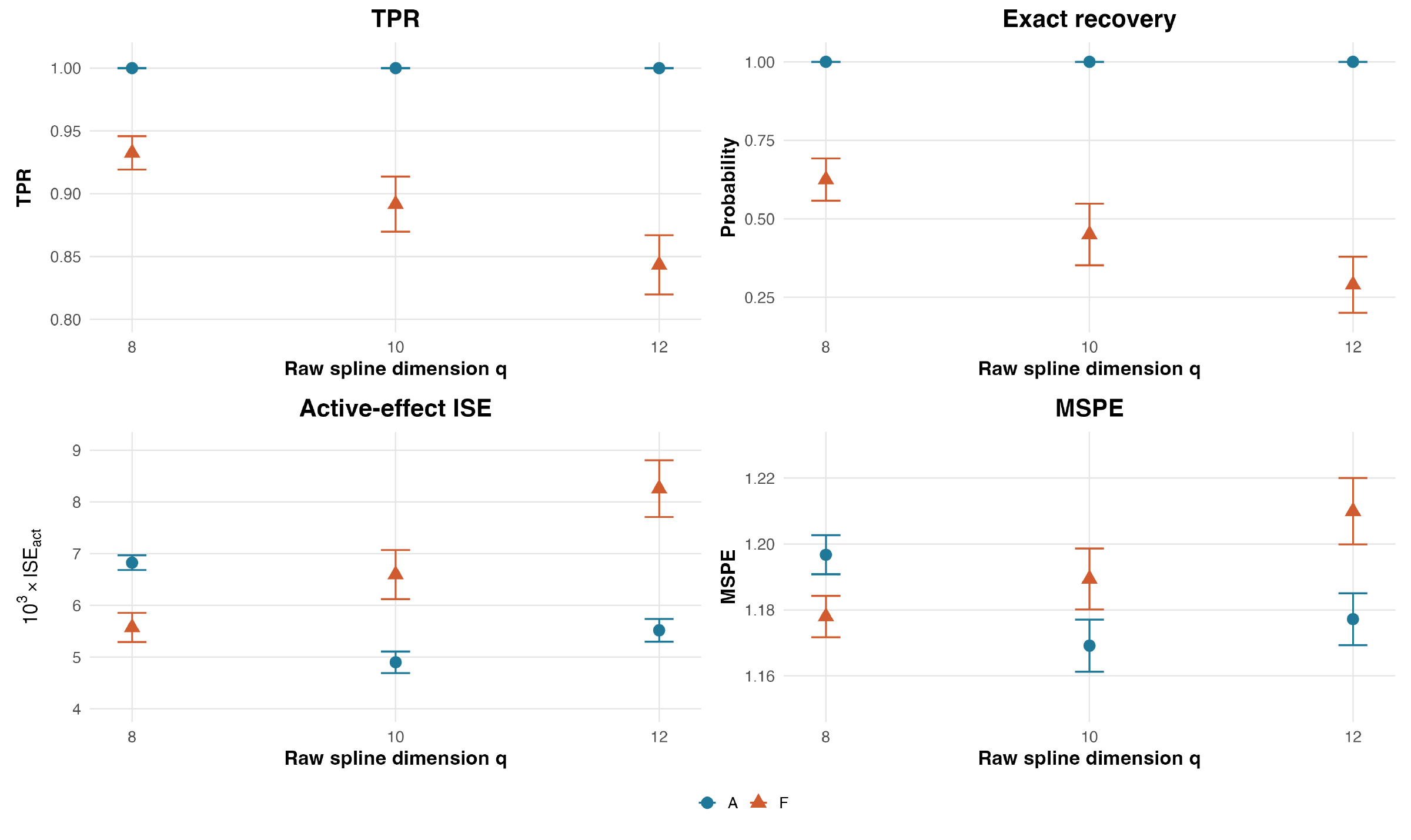}
\caption{Sensitivity to the spline dimension in Scenarios A and F for $(N,p)=(200,100)$.}
\label{fig:sim_q_sensitivity}
\end{figure}

\subsection{Summary}

Across the two sample-size designs and six data-generating scenarios, TV-Select attained the smallest active-effect ISE and MSPE in 11 of the 12 combinations of scenario and design and the highest exact-recovery probability in 10 of 12, while keeping FPR essentially at zero. For Scenarios A--E with $(N,p)=(200,100)$, exact recovery ranged from 0.915 to 1.000, together with high classification accuracy and stability. The structural decomposition removed unnecessary time variation while retaining the principal dynamic effects.

Compared with Group-Lasso, TV-Select retains high TPR while producing fewer false positives, lower active-effect ISE, and lower MSPE. Marginal-VC is less reliable under correlated and weak signals, while VC-Ridge incurs the variance associated with fitting every covariate as time varying. Curves averaged across all replications recover the principal shapes rather than depicting a favorable replication. The same pattern is present for $(N,p)=(100,50)$, and FPR remains controlled for $q=8,10,$ and $12$. Scenario F represents the weak-signal boundary, where Group-Lasso attains greater sensitivity by selecting more blocks. Although Group-Lasso has smaller RE in several settings, TV-Select has lower coefficient-function and prediction errors in most settings. Across these criteria, TV-Select combines false-positive control, support recovery, stable function estimation, and prediction accuracy.

\clearpage
\section{Real Data Analysis}\label{sec:6}

The analysis uses the \texttt{sleep-cassette} subset of the publicly available Sleep-EDF Expanded database on PhysioNet\footnote{\url{https://www.physionet.org/content/sleep-edfx/1.0.0/}} \citep{goldberger2000physiobank}. The polysomnographic recordings and their sleep-related scientific context were originally described by \citet{kemp2000analysis}. Repeated measurements from overnight polysomnography and manually scored sleep stages permit associations between physiological signals and slow-wave activity to vary over the course of sleep.

\subsection{Data and Feature Construction}

The analysis included 153 recordings from 78 subjects (41 women and 37 men; age range 25--101 years).  For each recording, we retained the interval from 30 minutes before the first scored sleep epoch to 30 minutes after the last scored sleep epoch and divided it into consecutive five-minute blocks. A block was retained when at least 80\% of its duration had a valid sleep-stage annotation. Block midpoints were normalized within recording to $t\in[0,1]$, giving 19,489 longitudinal observations.

The response was log-transformed delta-band power (0.5--4 Hz) from the EEG channel \texttt{Fpz-Cz},
\[
Y_{ij}=\log\{1+\operatorname{DeltaPower}_{ij}\},
\]
a continuous measure of slow-wave activity.  Nine block-level physiological features were candidates for time-varying effects: relative power in the theta, alpha, sigma, and beta bands from \texttt{Pz-Oz}; EOG root mean square (RMS) and line length; EMG RMS; and respiratory standard deviation and line length. Sex, age, sleep-stage depth, and within-block wake fraction were included as prespecified constant adjustments. Rectal temperature was retained for quality control but excluded from modeling because its scale was not comparable across recordings. Continuous variables were standardized using training-fold means and standard deviations, which were then applied without modification to the corresponding held-out subjects.

\subsection{Modeling and Evaluation}

For block $j$ from subject $i$, we fitted
\[
Y_{ij}
=
\beta_0+\sum_{k=1}^{9}X_{ijk}\beta_k(t_{ij})
+\boldsymbol Z_{ij}^{\top}\boldsymbol\gamma+\varepsilon_{ij},
\]
where $\boldsymbol Z_{ij}$ contains the four constant adjustments.  All methods used the same normalized effective cubic B-spline basis with raw dimension $q=8$ and the same preprocessing.  We compared TV-Select with VC-Ridge, Group-Lasso, and Marginal-VC, as in Section~\ref{sec:5}.

Prediction was evaluated by five-fold subject-level cross-validation repeated five times; every block from a subject was assigned to the same fold.  Within each outer training set, tuning parameters were selected by three-fold subject-level cross-validation using minimum mean subject RMSE.  Because the true coefficient functions are unknown, evaluation focused on held-out RMSE, MAE, $R^2$, and correlation.  Functional complexity was measured by the mean integrated squared second derivative over the nine candidate curves (Roughness).  For selection-capable methods, we also recorded the number of selected varying effects and the chance-adjusted pairwise Jaccard stability. Smaller Roughness indicates smoother curves, whereas larger adjusted stability indicates more reproducible selections after accounting for set size.

\subsection{Results}

Table~\ref{tab:real_results} summarizes prediction and structural behavior. VC-Ridge attained the smallest RMSE, but TV-Select was within 0.27\% of this value and within 0.15\% of Group-Lasso. Paired subject-level 95\% intervals for the RMSE, MAE, and $R^2$ differences between TV-Select and every comparator all contained zero. Held-out prediction was essentially indistinguishable among the four methods, with no material predictive cost from the structural and smoothness regularization used by TV-Select.

\begin{table}[h]
\centering
\caption{Subject-level prediction, curve regularity, and structural stability on Sleep-EDF.}
\label{tab:real_results}
\small
\renewcommand{\arraystretch}{1.10}
\setlength{\tabcolsep}{10pt}
\begin{tabular}{lcccc}
\toprule
\multicolumn{5}{c}{\textit{Panel A: held-out prediction}}\\
\cmidrule(lr){1-5}
Method & RMSE & MAE & $R^2$ & Correlation\\
\midrule
TV-Select    & 0.7617 (0.0217) & 0.6216 (0.0185) & 0.2709 (0.0444) & 0.7027 (0.0151)\\
VC-Ridge     & 0.7596 (0.0216) & 0.6206 (0.0183) & 0.2749 (0.0441) & 0.7053 (0.0151)\\
Group-Lasso  & 0.7605 (0.0217) & 0.6211 (0.0183) & 0.2736 (0.0439) & 0.7041 (0.0152)\\
Marginal-VC  & 0.7628 (0.0217) & 0.6225 (0.0184) & 0.2692 (0.0440) & 0.7005 (0.0151)\\
\bottomrule
\end{tabular}

\vspace{0.8em}

\begin{tabular}{lccc}
\toprule
\multicolumn{4}{c}{\textit{Panel B: functional complexity and selection}}\\
\cmidrule(lr){1-4}
Method & Roughness & No.\ varying effects & Adjusted stability\\
\midrule
TV-Select    & 43.65 (2.21)   & 7.12 (0.23) & 0.413 (0.047)\\
VC-Ridge     & 49.62 (5.01)   & 9.00        & N/A\\
Group-Lasso  & 90.73 (10.02)  & 8.52 (0.05) & 0.286 (0.001)\\
Marginal-VC  & 147.34 (40.54) & 6.00 (0.30) & 0.435 (0.083)\\
\bottomrule
\end{tabular}

\vspace{0.4em}

\begin{minipage}{0.95\textwidth}
\footnotesize
\textit{Note:} Parenthetical values are standard errors across repeated outer splits. Roughness is averaged over the nine common candidate curves. VC-Ridge retains all nine varying blocks and therefore has no selection-stability measure.
\end{minipage}
\end{table}

The principal difference was functional regularity. TV-Select had the lowest mean Roughness, 12\% below VC-Ridge, 52\% below Group-Lasso, and 70\% below Marginal-VC. Figure~\ref{fig:real_roughness} shows that the lower roughness persisted across repeated outer splits rather than being driven by one partition. TV-Select also selected fewer varying effects than Group-Lasso (7.12 versus 8.52) and had higher chance-adjusted stability (0.413 versus 0.286). Marginal-VC selected slightly fewer effects and had similar adjusted stability, but its curves were more than three times as rough and its predictive performance was slightly weaker. TV-Select combines near-equivalent prediction with greater structural parsimony than Group-Lasso and smoother temporal effects than all three comparators.

\begin{figure}[h]
\centering
\includegraphics[width=0.75\textwidth]{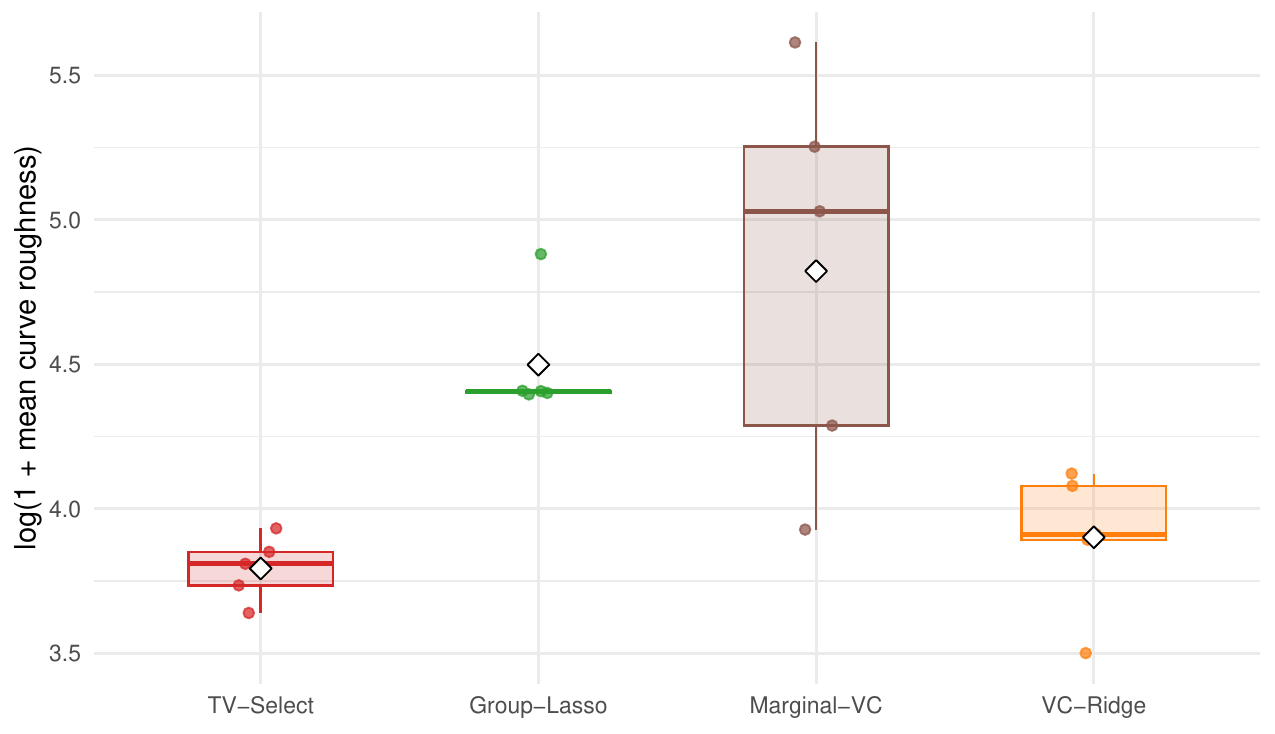}
\caption{Curve roughness across repeated subject-level cross-validation splits.}
\label{fig:real_roughness}
\end{figure}

The full-data TV-Select fit selected five time-varying effects: alpha, sigma, and beta power, EOG line length, and respiratory line length.  Alpha, sigma, and EOG line length were selected in every outer fit, while respiratory line length was selected in 84\% of fits.  Figure~\ref{fig:real_curves} displays the corresponding full-data coefficient estimates.  TV-Select preserves the main temporal trends while suppressing the additional oscillations visible for Marginal-VC and, to a lesser extent, Group-Lasso and VC-Ridge.

The fitted profiles also have plausible physiological interpretations. Alpha power was negatively associated with delta activity throughout the night, with the strongest negative association early in normalized sleep time. The sigma association changed from negative to positive, consistent with a changing relationship between spindle-band and slow-wave activity. EOG line length showed an increasingly positive association through the middle of the night, whereas respiratory line length remained negative and became more pronounced late in the night. The beta effect was comparatively small. These curves describe conditional associations after adjustment for demographics and sleep-stage composition; they are not causal effects or sleep-stage contrasts. They identify smooth, interpretable temporal patterns without sacrificing held-out predictive performance.

\begin{figure}[!tbp]
\centering
\includegraphics[width=0.99\textwidth]{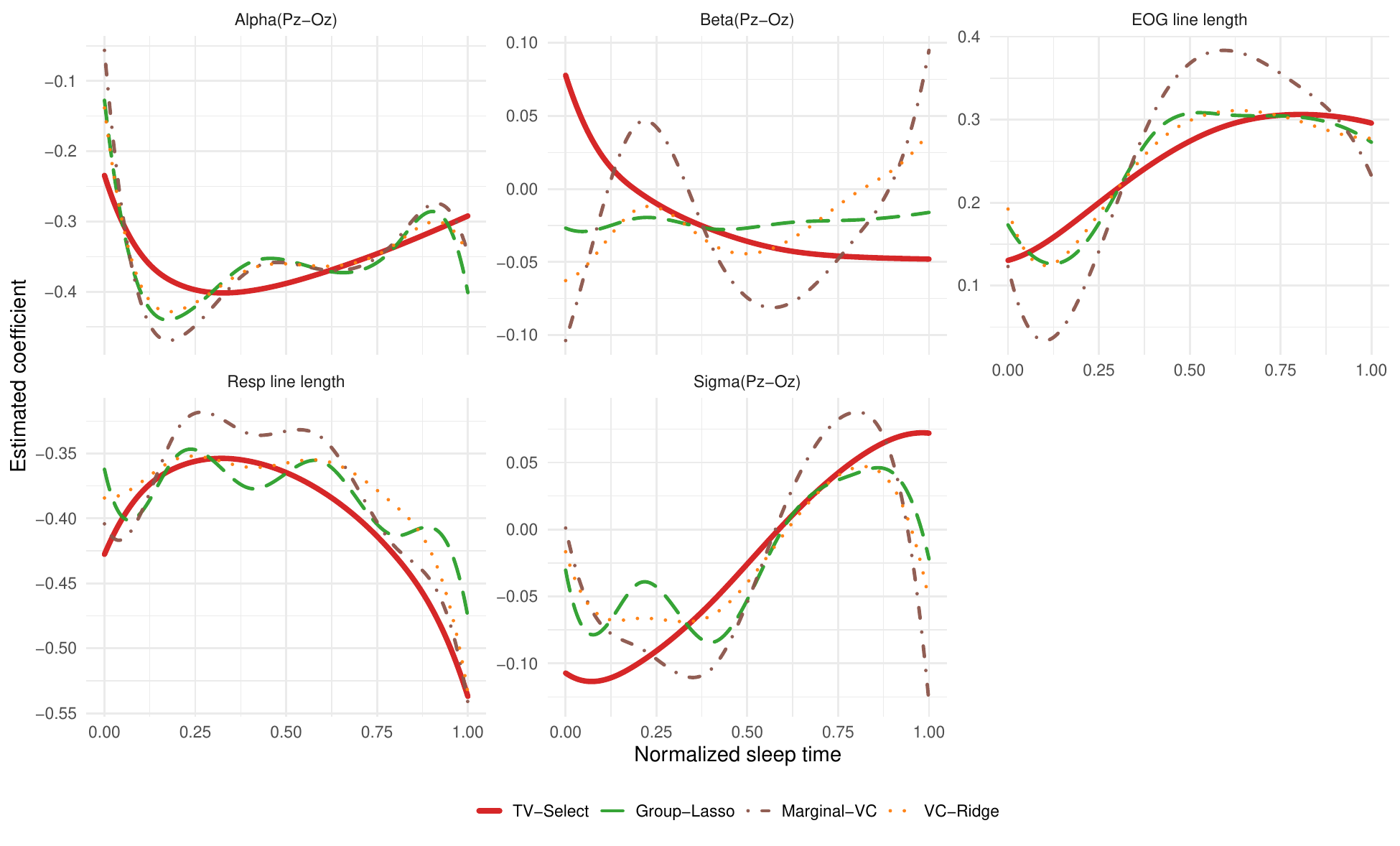}
\caption{Coefficient curves for effects selected as time varying by TV-Select.}
\label{fig:real_curves}
\end{figure}
\FloatBarrier

\section{Conclusion}\label{sec:7}
TV-Select distinguishes zero, constant, and time-varying effects within a single longitudinal varying-coefficient model. The mean and centered deviation define the three classes, and the normalized effective basis removes the rank deficiency created by centering. Group sparsity and roughness regularization are combined in a convex criterion, followed by distinct refits for curve estimation, classification, and inference. The theoretical results cover function estimation, recovery of the varying set, separation of zero and constant effects, and cluster-robust oracle inference for fixed contrasts.

Across the simulations, TV-Select maintains low false-positive rates and competitive estimation and prediction under dependence, heavy-tailed errors, and weak signals. The Sleep-EDF analysis produces smoother and more parsimonious effect trajectories than Group-Lasso with comparable held-out prediction. The improved structural interpretation does not require a substantial predictive loss. Further work is needed for penalized high-dimensional constant effects, covariance-aware efficiency \citep{bai2023nvcssl}, and simultaneous inference for coefficient curves.

\section*{Disclosure statement}
No potential conflict of interest was reported by the author(s).


\begingroup
\bibliographystyle{erae}
\bibliography{references}
\endgroup

\newpage
\appendix
\section*{Proofs}

\subsection*{Notation and profiled criterion}
Let $\gamma_0=(\beta_0,\boldsymbol\mu_0^\top)^\top$.  In normalized effective coordinates, the spline approximation in Assumption~\ref{ass:A2} gives
\begin{equation}\label{eq:app_model_new}
\mathbf y=\bar{\mathbf X}\gamma_0+
\mathbf Z\boldsymbol\theta_0^\ast+
\boldsymbol\varepsilon+\mathbf r.
\end{equation}
For fixed $\boldsymbol\theta$, least squares in $\gamma$ yields the projection $M_X$.  Hence the $\boldsymbol\theta$ component of every minimizer of \eqref{eq:objective} minimizes the convex profiled criterion
\begin{equation}\label{eq:profiled_obj_new}
\mathcal Q_p(\boldsymbol\theta)=
\frac{1}{2n}\|M_X(\mathbf y-\mathbf Z\boldsymbol\theta)\|_2^2
+\lambda_1\|\boldsymbol\theta\|_{2,1}
+\lambda_2\boldsymbol\theta^\top\Omega_p\boldsymbol\theta.
\end{equation}

\subsection*{Proof of Proposition~\ref{prop:algorithm_convergence}}
\begin{proof}
Fix all blocks except $\boldsymbol\theta_k$ and write the corresponding smooth part as $\ell_k$.  Lipschitz continuity of its gradient and the optimality of the proximal step imply
\[
\mathcal Q(\boldsymbol\theta^{(m)})
-\mathcal Q(\boldsymbol\theta^{(m+1)})
\ge
\left\{\frac{1}{2\alpha_k}-\frac{L_k}{2}\right\}
\|\boldsymbol\theta_k^{(m+1)}
-\boldsymbol\theta_k^{(m)}\|_2^2.
\]
Because $\alpha_k\le(1-\nu)L_k^{-1}$, the coefficient on the right is at least $\nu/(2\alpha_k)$ and hence is uniformly positive.  The exact least-squares update of $\gamma$ also cannot increase the objective. Summing the descent inequalities over iterations shows that every block increment of $\boldsymbol\theta$ is square summable.  Moreover, because $\bar{\mathbf X}$ has full column rank, the exact least-squares update is the affine map
\[
\gamma(\boldsymbol\theta)
=
(\bar{\mathbf X}^{\top}\bar{\mathbf X})^{-1}
\bar{\mathbf X}^{\top}
(\mathbf y-\mathbf Z\boldsymbol\theta).
\]
Consequently, for the finite constant $C_\gamma=\|(\bar{\mathbf X}^{\top}\bar{\mathbf X})^{-1}\bar{\mathbf X}^{\top}\mathbf Z\|_{\rm op}$,
\[
\|\gamma^{(m+1)}-\gamma^{(m)}\|_2
\le C_\gamma
\|\boldsymbol\theta^{(m+1)}-\boldsymbol\theta^{(m)}\|_2
\longrightarrow0.
\]
Thus the unpenalized-block increments also vanish.  The bounded-level-set assumption provides an accumulation point.  Since every within-sweep block increment vanishes, any subsequential limit of intermediate iterates coincides with the associated sweep-end limit.  Along a convergent subsequence, continuity of the proximal map gives zero proximal-gradient residual for every varying block, while the affine least-squares map gives the normal equations for $\gamma$.  These are exactly the KKT conditions of the convex criterion \eqref{eq:objective}, so every accumulation point is a global minimizer.

The objective values converge to the common minimum attained by all accumulation points.  The squared-error term is strictly convex in the fitted value, so two minimizers cannot have different fitted values; otherwise their midpoint would have a strictly smaller objective.  Hence the fitted values converge.  If the coefficient minimizer is unique in the normalized effective coordinates, the bounded sequence has only one accumulation point and therefore converges to it.
\end{proof}

\subsection*{Auxiliary Lemmas}
\begin{lemma}[Uniform block score]\label{lem:score_bounds_new}
Under Assumptions~\ref{ass:A1}--\ref{ass:A4},
\begin{align}
\left\|\frac{\tilde{\mathbf Z}^\top
\boldsymbol\varepsilon}{n}\right\|_{2,\infty}
&=O_p(\zeta_n),\label{eq:score_z}\\
\frac{\|P_X\boldsymbol\varepsilon\|_2^2}{n}
&=O_p(p/n).\label{eq:score_x}
\end{align}
More explicitly, on the design event in Assumption~\ref{ass:A1}, there are constants $c,C>0$ such that, for every $t\ge0$,
\begin{equation}\label{eq:score_tail}
\Pr\left\{
\left\|\frac{\tilde{\mathbf Z}^{\top}
\boldsymbol\varepsilon}{n}\right\|_{2,\infty}
>
C\sqrt{\frac{q_e+\log p+t}{n}}
\ \middle|\ \mathcal D
\right\}
\le 2e^{-ct}.
\end{equation}
Moreover, $\|\tilde{\mathbf Z}^\top M_X\mathbf r/n\|_{2,\infty}=b_n$ and $\|\mathbf r\|_2^2/n=O(s_vq_e^{-2r})$.
\end{lemma}

\begin{proof}
Because $\tilde{\mathbf Z}_k=M_X\mathbf Z_k$ and $M_X$ is symmetric idempotent, $\tilde{\mathbf Z}_k^\top M_X\boldsymbol\varepsilon=\tilde{\mathbf Z}_k^\top\boldsymbol\varepsilon$. Partition $\tilde{\mathbf Z}_k$ by subject.  Conditional on the design, for fixed $\|u\|_2=1$,
\[
\frac1n u^\top\tilde{\mathbf Z}_k^\top\boldsymbol\varepsilon =\frac1n\sum_{i=1}^N
u^\top\tilde{\mathbf Z}_{k,i}^\top\boldsymbol\varepsilon_i
\]
is a sum of independent sub-Gaussian variables. Its squared sub-Gaussian norm is bounded by
\[
\frac{K_\varepsilon^2}{n^2}
\sum_{i=1}^N\|\tilde{\mathbf Z}_{k,i}u\|_2^2
\le \frac{K_\varepsilon^2K_Z}{n}.
\]
A $1/2$-net of the unit sphere in $\mathbb R^{q_e}$ has at most $5^{q_e}$ elements.  Applying the sub-Gaussian tail bound on this net and then taking a union bound over $p$ blocks gives, for a sufficiently large constant $C$,
\[
\Pr\left\{
\max_{k\le p}
\left\|\frac{\tilde{\mathbf Z}_k^\top
\boldsymbol\varepsilon}{n}\right\|_2
>C\sqrt{\frac{q_e+\log p+t}{n}}
\ \middle|\ \mathcal D
\right\}
\le 2e^{-ct}.
\]
This proves \eqref{eq:score_tail}, and taking $t$ fixed proves \eqref{eq:score_z}.  For \eqref{eq:score_x}, conditional covariance boundedness and $\operatorname{rank}(P_X)\le p+1$ give
\[
\mathbb{E}(\|P_X\boldsymbol\varepsilon\|_2^2\mid\mathcal D)
=\operatorname{tr}\{P_X\operatorname{Var}
(\boldsymbol\varepsilon\mid\mathcal D)\}
\le C(p+1).
\]
Markov's inequality proves the claim. The two approximation statements are Assumption~\ref{ass:A2}.
\end{proof}

\subsection*{Proof of Theorem~\ref{thm:T1_rate}}
\begin{proof}
Write $\Delta=\hat{\boldsymbol\theta}-\boldsymbol\theta_0^\ast$ and
\[
h=\frac{\tilde{\mathbf Z}^\top
(\boldsymbol\varepsilon+M_X\mathbf r)}{n},\qquad
d_S=2\lambda_2\Omega_S\boldsymbol\theta_{0S}^\ast.
\]
Lemma~\ref{lem:score_bounds_new} and Assumption~\ref{ass:A4} imply that, on an event whose probability tends to one,
\[
\|h\|_{2,\infty}\le\lambda_1/4,\qquad
\|d_S\|_{2,\infty}\le\lambda_1/4,
\]
because $A_n\to\infty$ and $c_0$ is sufficiently small. Optimality of \eqref{eq:profiled_obj_new} at $\hat{\boldsymbol\theta}$, followed by expansion of the quadratic roughness term, gives
\begin{align}
\frac{\|\tilde{\mathbf Z}\Delta\|_2^2}{2n}
+\lambda_2\Delta^\top\Omega_p\Delta
+\lambda_1\{
\|\boldsymbol\theta_0^\ast+\Delta\|_{2,1}
-\|\boldsymbol\theta_0^\ast\|_{2,1}\}
\le \langle h,\Delta\rangle-\langle d_S,\Delta_S\rangle.
\label{eq:basic_profile}
\end{align}
By block Hölder inequality and decomposability,
\[
\langle h,\Delta\rangle
\le\frac{\lambda_1}{4}\|\Delta\|_{2,1},\qquad
-\langle d_S,\Delta_S\rangle
\le\frac{\lambda_1}{4}\|\Delta_S\|_{2,1},
\]
and
\[
\|\boldsymbol\theta_0^\ast+\Delta\|_{2,1}
-\|\boldsymbol\theta_0^\ast\|_{2,1}
\ge\|\Delta_{S^c}\|_{2,1}-\|\Delta_S\|_{2,1}.
\]
Dropping the nonnegative roughness quadratic in \eqref{eq:basic_profile} therefore yields
\begin{equation}\label{eq:cone_and_basic}
\frac{\|\tilde{\mathbf Z}\Delta\|_2^2}{2n}
+\frac{\lambda_1}{2}\|\Delta_{S^c}\|_{2,1}
\le\frac{3\lambda_1}{2}\|\Delta_S\|_{2,1}.
\end{equation}
In particular, $\Delta\in\mathcal C(S,3)$.  Since $\|\Delta_S\|_{2,1}\le\sqrt{s_v}\|\Delta_S\|_2$, \eqref{eq:compatibility} and \eqref{eq:cone_and_basic} imply
\[
\frac{\|\tilde{\mathbf Z}\Delta\|_2}{\sqrt n}
\le \frac{3\sqrt{s_v}}{\kappa_Z}\lambda_1,\qquad
\|\Delta_S\|_2
\le \frac{3\sqrt{s_v}}{\kappa_Z^2}\lambda_1.
\]
Substitution back into \eqref{eq:cone_and_basic} also gives $\|\Delta_{S^c}\|_{2,1}=O_p(s_v\lambda_1)$, proving \eqref{eq:t1_profile_prediction}--\eqref{eq:t1_parameter}.

The least-squares normal equation for $\hat\gamma$ gives the exact orthogonal decomposition
\begin{equation*}
\bar{\mathbf X}(\hat\gamma-\gamma_0)+\mathbf Z\Delta
=P_X(\boldsymbol\varepsilon+\mathbf r)
+M_X\mathbf Z\Delta.
\end{equation*}
The two terms on the right lie in orthogonal subspaces.  By Lemma~\ref{lem:score_bounds_new}, Assumption~\ref{ass:A2}, and \eqref{eq:t1_profile_prediction},
\[
\frac1n\|\bar{\mathbf X}(\hat\gamma-\gamma_0)+\mathbf Z\Delta\|_2^2
=O_p\left(\frac pn+s_vq_e^{-2r}+s_v\lambda_1^2\right).
\]
Applying \eqref{eq:joint_identifiability} to $(\hat\gamma-\gamma_0,\Delta)$ proves
\eqref{eq:t1_joint_parameter}.

It remains to prove the maximum bound.  The active KKT equations are
\begin{equation}\label{eq:active_kkt_t1}
G_S\Delta_S
=h_S-\lambda_1\hat u_S-d_S-H_{S S^c}\Delta_{S^c},
\qquad \|\hat u_S\|_{2,\infty}\le1.
\end{equation}
Assumption~\ref{ass:A3} and
$\|\Delta_{S^c}\|_{2,1}=O_p(s_v\lambda_1)$ imply
\[
\|H_{S S^c}\Delta_{S^c}\|_{2,\infty}
\le C_H\|\Delta_{S^c}\|_{2,1}
=O_p(\lambda_1)
\]
when $s_v=O(1)$; here the roughness matrix is block diagonal, so the off-diagonal blocks of $H$ are exactly the corresponding empirical Gram blocks.  Every term on the right of \eqref{eq:active_kkt_t1} is therefore $O_p(\lambda_1)$ in block maximum norm, and \eqref{eq:block_inverse_condition} yields
\[
\|\Delta_S\|_{2,\infty}
\le C_G\|h_S-\lambda_1\hat u_S-d_S-
H_{SS^c}\Delta_{S^c}\|_{2,\infty}
=O_p(\lambda_1).
\]
This proves
\eqref{eq:t1_max}.  The bound
\[
\|\hat g_k-g_{0k}\|_{L_2}
\le
\|\boldsymbol C(\cdot)^\top
(\hat{\boldsymbol\theta}_k-\boldsymbol\theta_{0k}^\ast)\|_{L_2}
+K_gq_e^{-r}
=\|\hat{\boldsymbol\theta}_k-\boldsymbol\theta_{0k}^\ast\|_2
+K_gq_e^{-r},
\]
which proves \eqref{eq:t1_function}.
\end{proof}

\subsection*{Proof of Theorem~\ref{thm:T2_select}}
\begin{proof}
We use a primal--dual witness argument for \eqref{eq:profiled_obj_new}.  Set the inactive blocks to zero and let $\tilde{\boldsymbol\theta}_S$ be the unique restricted minimizer.  Its active KKT equation is
\begin{equation*}
G_S(\tilde{\boldsymbol\theta}_S-\boldsymbol\theta_{0S}^\ast)
=h_S-d_S-\lambda_1\tilde u_S,
\qquad \|\tilde u_S\|_{2,\infty}\le1.
\end{equation*}
The score lemma, smoothing-bias condition, and active inverse condition give
\[
\|\tilde{\boldsymbol\theta}_S-\boldsymbol\theta_{0S}^\ast
\|_{2,\infty}
\le C_G\{\|h_S\|_{2,\infty}
+\|d_S\|_{2,\infty}+\lambda_1\}
\le (C_G+o_p(1))\lambda_1.
\]
Consequently,
\[
\min_{k\in S}\|\tilde{\boldsymbol\theta}_k\|_2
\ge
\min_{k\in S}\|\boldsymbol\theta_{0k}^\ast\|_2
-(C_G+o_p(1))\lambda_1>0
\]
with probability tending to one by the beta-min condition.  Thus every restricted active block is nonzero.

For $k\notin S$, the negative smooth gradient at the restricted solution is
\begin{align*}
R_k
&=h_k-H_{kS}
(\tilde{\boldsymbol\theta}_S-\boldsymbol\theta_{0S}^\ast)\\
&=
\{h_k-H_{kS}G_S^{-1}(h_S-d_S)\}
+\lambda_1H_{kS}G_S^{-1}\tilde u_S.
\end{align*}
The second term is at most $(1-\eta)\lambda_1$ by \eqref{eq:irrepresentable}.  The same operator bound, the clustered score lemma, $A_n\to\infty$, and $b_n+\delta_n=o(\lambda_1)$ imply
\[
\begin{split}
\max_{k\notin S}
\|h_k-H_{kS}G_S^{-1}(h_S-d_S)\|_2
&\le \|h\|_{2,\infty}
 +(1-\eta)\|h_S-d_S\|_{2,\infty}\\
&=o_p(\lambda_1).
\end{split}
\]
Indeed, $\|h\|_{2,\infty}=O_p(\zeta_n)+b_n=o_p(\lambda_1)$ and $\|d_S\|_{2,\infty}\le\delta_n=o(\lambda_1)$.  Hence the last display is at most $\eta\lambda_1/2$ with probability tending to one, and $\max_{k\notin S}\|R_k\|_2\le(1-\eta/2)\lambda_1<\lambda_1$. This is strict dual feasibility.  Taking inactive subgradients $R_k/\lambda_1$ makes the restricted primal solution satisfy the full KKT system.

To establish uniqueness, let $\boldsymbol\theta'$ be any other minimizer and use the KKT subgradient just constructed at $(\tilde{\boldsymbol\theta}_S^\top,0^\top)^\top$.  The subgradient inequality for an inactive block is strict whenever $\boldsymbol\theta'_k\ne0$, because
\[
\|\boldsymbol\theta'_k\|_2-
\left\langle R_k/\lambda_1,\boldsymbol\theta'_k\right\rangle
\ge
\{1-\|R_k/\lambda_1\|_2\}\|\boldsymbol\theta'_k\|_2>0.
\]
Equality of the objective values therefore forces $\boldsymbol\theta'_{S^c}=0$. On the restricted space, $G_S\succ0$ makes the smooth part strictly convex, so $\boldsymbol\theta'_S=\tilde{\boldsymbol\theta}_S$.  Thus the full minimizer is unique.  Together with the beta-min result, this proves $\Pr(\widehat{\mathcal S}_{\mathrm{vary}}=S)\to1$.
\end{proof}

\subsection*{Proof of Corollary~\ref{cor:smooth_refit_rate}}
\begin{proof}
Let $\mathcal A_n=\{\widehat S=S\}$; by Theorem~\ref{thm:T2_select}, $\Pr(\mathcal A_n)\to1$.  Work on $\mathcal A_n$ and write
\[
\alpha_0=(\gamma_0^\top,
(\boldsymbol\theta_{0S}^\ast)^\top)^\top,\qquad
P_F=\operatorname{blockdiag}
\{0_{p+1},\,2\hat\lambda_2(I_{s_v}\otimes\Omega)\}.
\]
The normal equation for \eqref{eq:smooth_refit} and \eqref{eq:app_model_new} give the exact identity
\begin{equation}\label{eq:smooth_refit_expansion}
\hat\alpha^{\,F}-\alpha_0
=
\left\{
\frac{(D_S^{\rm all})^\top D_S^{\rm all}}{n}+P_F
\right\}^{-1}
\left[
\frac{(D_S^{\rm all})^\top
(\boldsymbol\varepsilon+\mathbf r)}{n}
-P_F\alpha_0
\right].
\end{equation}
The assumed lower eigenvalue bound and $P_F\succeq0$ imply that the inverse in \eqref{eq:smooth_refit_expansion} has uniformly bounded operator norm.  Conditional sub-Gaussian concentration, bounded cluster size, and $d_F=o(n)$ yield
\[
\left\|
\frac{(D_S^{\rm all})^\top\boldsymbol\varepsilon}{n}
\right\|_2
=O_p\{\sqrt{d_F/n}\}.
\]
Moreover, the upper eigenvalue bound and Assumption~\ref{ass:A2} give
\[
\left\|
\frac{(D_S^{\rm all})^\top\mathbf r}{n}
\right\|_2
\le
\lambda_{\max}^{1/2}\!\left\{
\frac{(D_S^{\rm all})^\top D_S^{\rm all}}{n}
\right\}
\frac{\|\mathbf r\|_2}{\sqrt n}
=O(\sqrt{s_v}\,q_e^{-r}),
\]
whereas $\|P_F\alpha_0\|_2\le\sqrt{s_v}\,\delta_{F,n}$. Substitution in \eqref{eq:smooth_refit_expansion} and squaring prove \eqref{eq:smooth_refit_parameter_rate}.

For each $k\in S$, normalized effective coordinates and the spline approximation bound imply
\[
\|\hat g_k^{\,F}-g_{0k}\|_{L_2}^2
\le
2\|\hat{\boldsymbol\theta}_k^{\,F}
-\boldsymbol\theta_{0k}^\ast\|_2^2
+2K_g^2q_e^{-2r}.
\]
Summing this inequality proves \eqref{eq:smooth_refit_function_rate}.  The upper eigenvalue bound for $D_S^{\rm all}$ converts the parameter bound into the same in-sample prediction order.  Since $\Pr(\mathcal A_n^c)\to0$, all conclusions hold unconditionally.
\end{proof}

\subsection*{Proof of Lemma~\ref{lem:classification_refit_rate}}
\begin{proof}
Write $D=D_S^{\rm all}$ and $d_R=p+1+s_vq_e$ for its number of columns.  The upper eigenvalue bound in Assumption~\ref{ass:A6} implies
\[
\max_{\ell\le d_R}\frac{\|D_\ell\|_2^2}{n}
\le \lambda_{\max}(D^\top D/n)
\le \kappa_R^{-1}.
\]
Conditional on the design, each coordinate of $D^\top\boldsymbol\varepsilon/n$ is a sum of independent subject-level sub-Gaussian variables.  By Assumption~\ref{ass:A1}, its sub-Gaussian norm is at most $K_\varepsilon\|D_\ell\|_2/n\le Cn^{-1/2}$.  Hence, for constants $c_1,c_2>0$,
\[
\Pr\!\left(
\left\|\frac{D^\top\boldsymbol\varepsilon}{n}\right\|_\infty>u
\ \middle|\ \mathcal D
\right)
\le 2d_R\exp(-c_1nu^2)
\]
for $0<u<c_2$.  Taking
$u=C\sqrt{\log(d_R)/n}$ and using
$\log d_R=O(\log p)$ proves \eqref{eq:refit_score_bound}.

For the approximation term, Cauchy--Schwarz, the same column-norm bound, and Assumption~\ref{ass:A2} give
\[
\left\|\frac{D^\top\mathbf r}{n}\right\|_\infty
\le
\max_{\ell\le d_R}\frac{\|D_\ell\|_2}{\sqrt n}
\frac{\|\mathbf r\|_2}{\sqrt n}
\le
\kappa_R^{-1/2}O(\sqrt{s_v}\,q_e^{-r})
=O(q_e^{-r}),
\]
where the last equality uses $s_v=O(1)$.  This proves \eqref{eq:refit_remainder_bound}.

The oracle least-squares refit has the exact expansion
\[
\begin{pmatrix}
\hat\gamma^{\,R}-\gamma_0\\
\hat{\boldsymbol\theta}^{\,R}_S-\boldsymbol\theta_{0S}^\ast
\end{pmatrix}
=B_S\frac{D^\top(\boldsymbol\varepsilon+\mathbf r)}{n}.
\]
For every constant-effect row $\ell\in\mathcal I_\mu$,
\[
\left|e_\ell^\top B_S
\frac{D^\top(\boldsymbol\varepsilon+\mathbf r)}{n}\right|
\le
\|e_\ell^\top B_S\|_1
\left\|\frac{D^\top(\boldsymbol\varepsilon+\mathbf r)}{n}\right\|_\infty.
\]
Taking the maximum over $\mathcal I_\mu$, applying the row bound in Assumption~\ref{ass:A6}, and using \eqref{eq:refit_score_bound}--\eqref{eq:refit_remainder_bound} proves \eqref{eq:refit_stability}.
\end{proof}

\subsection*{Proof of Corollary~\ref{cor:mu_threshold}}
\begin{proof}
On the event $\widehat{\mathcal S}_{\mathrm{vary}}=S$, which has probability tending to one, the refit in \eqref{eq:structure_refit} uses $D_S^{\rm all}=[\bar{\mathbf X},\mathbf Z_S]$.  Its exact least-squares expansion and \eqref{eq:app_model_new} give
\[
\begin{pmatrix}
\hat\gamma^{\,R}-\gamma_0\\
\hat{\boldsymbol\theta}^{\,R}_S-\boldsymbol\theta_{0S}^\ast
\end{pmatrix}
=
\left\{(D_S^{\rm all})^\top D_S^{\rm all}/n\right\}^{-1}
\frac{(D_S^{\rm all})^\top
(\boldsymbol\varepsilon+\mathbf r)}{n}.
\]
Lemma~\ref{lem:classification_refit_rate} therefore gives
$\|\hat{\boldsymbol\mu}^{\,R}-\boldsymbol\mu_0\|_\infty
=O_p(r_{\mu,n})$.
Since $r_{\mu,n}=o(\tau_N)$, every zero coordinate is below
$\tau_N$ with probability tending to one.  For
$k\in\mathcal S_{\mathrm{const}}$,
\[
|\hat\mu_k^{\,R}|
\ge|\mu_{0k}|-
|\hat\mu_k^{\,R}-\mu_{0k}|
>\tau_N
\]
with probability tending to one by the stated beta-min condition. Uniformity of the sup-norm bound proves simultaneous recovery of the two sets.
\end{proof}

\subsection*{Proof of Theorem~\ref{thm:T3_oracle}}
\begin{proof}
Let $\mathcal A_n$ be the event that both structural sets are recovered. Theorem~\ref{thm:T2_select} and Corollary~\ref{cor:mu_threshold} imply $\Pr(\mathcal A_n)\to1$.  On $\mathcal A_n$, the selected design equals $D^{or}=[\mathbf X_c,W]$.  For arbitrary selected-model starting values $(\hat{\boldsymbol\mu}^{\,P}_c,\hat\eta^P)$, the one-step estimator obeys
\begin{align*}
\hat{\boldsymbol\mu}^{\,DB}_c
&=\hat{\boldsymbol\mu}^{\,P}_c+
A_{c,n}^{-1}\frac{\tilde{\mathbf X}_c^\top
\{\mathbf y-\mathbf X_c\hat{\boldsymbol\mu}^{\,P}_c-W\hat\eta^P\}}{n}\\
&=A_{c,n}^{-1}\frac{\tilde{\mathbf X}_c^\top\mathbf y}{n}
=A_{c,n}^{-1}\frac{\tilde{\mathbf X}_c^\top M_W\mathbf y}{n},
\end{align*}
because $\tilde{\mathbf X}_c^\top W=0$ and $\tilde{\mathbf X}_c^\top\mathbf X_c/n=A_{c,n}$. The Frisch--Waugh--Lovell theorem identifies the last expression with the $\mathbf X_c$ coefficient in the oracle least-squares refit.  This proves \eqref{eq:db_oracle_equivalence}.

Using $\tilde{\mathbf X}_c=M_W\mathbf X_c$ and \eqref{eq:app_model_new} gives the exact expansion
\begin{equation}\label{eq:oracle_expansion}
\hat{\boldsymbol\mu}^{\,DB}_c-\boldsymbol\mu_{0,c}
=A_{c,n}^{-1}
\frac{\tilde{\mathbf X}_c^\top\boldsymbol\varepsilon}{n}
+A_{c,n}^{-1}
\frac{\tilde{\mathbf X}_c^\top M_W\mathbf r}{n}.
\end{equation}
The sufficient condition stated after \eqref{eq:t3_bias} follows because $\|A_{c,n}^{-1}\|_{\rm op}=O_p(1)$, $\|\tilde{\mathbf X}_c\|_{\rm op}/\sqrt n=O_p(1)$, and $\|M_W\|_{\rm op}\le1$ imply
\[
\sqrt n\left\|
A_{c,n}^{-1}\frac{\tilde{\mathbf X}_c^\top M_W\mathbf r}{n}
\right\|_2
\le O_p(1)\|\mathbf r\|_2
=O_p(\sqrt{n s_v}\,q_e^{-r}).
\]
Thus $\sqrt{n s_v}\,q_e^{-r}\to0$ implies \eqref{eq:t3_bias}. The second term in \eqref{eq:oracle_expansion} is therefore $o_p(n^{-1/2})$.  For the first, write
\[
\frac{\tilde{\mathbf X}_c^\top\boldsymbol\varepsilon}{\sqrt n}
=\sum_{i=1}^N
\frac{\tilde{\mathbf X}_{c,i}^\top\boldsymbol\varepsilon_i}{\sqrt n}.
\]
Conditional on the design, the summands are independent across subjects, have mean zero, and have covariance sum $\Gamma_{c,n}$. The Cramér--Wold device and \eqref{eq:cluster_lindeberg} give the multivariate Lindeberg--Feller limit
\[
\Gamma_{c,n}^{-1/2}
\frac{\tilde{\mathbf X}_c^\top\boldsymbol\varepsilon}{\sqrt n}
\xrightarrow{d}N(0,I).
\]
Combining this limit with \eqref{eq:t3_A}, \eqref{eq:t3_Gamma}, \eqref{eq:oracle_expansion}, and Slutsky's theorem proves the asserted normal limit on $\mathcal A_n$.  Since $\Pr(\mathcal A_n^c)\to0$, the same limit holds unconditionally.

For sandwich consistency, set $u_i=\tilde{\mathbf X}_{c,i}^\top\boldsymbol\varepsilon_i$ and $\hat u_i=\tilde{\mathbf X}_{c,i}^\top\hat{\boldsymbol\varepsilon}_i$.  Independence across subjects, bounded fourth moments, bounded cluster size, and the bounded design condition imply
\begin{equation}\label{eq:sandwich_lln}
\frac1n\sum_{i=1}^N
\{u_iu_i^\top-\mathbb{E}(u_iu_i^\top\mid\mathcal D)\}=o_p(1).
\end{equation}
To see this directly, conditional independence and the fixed dimension of $S_c$ give
\[
\mathbb{E}\!\left[
\left\|\frac1n\sum_{i=1}^N
\{u_iu_i^\top-\mathbb{E}(u_iu_i^\top\mid\mathcal D)\}\right\|_F^2
\middle|\mathcal D\right]
\le \frac{K}{n^2}\sum_{i=1}^N1
=O(n^{-1}),
\]
because $N\asymp n$.  Conditional Chebyshev's inequality proves \eqref{eq:sandwich_lln}.

Let $P_{or}$ be the projection onto the oracle design.  Oracle residuals satisfy
\[
\hat{\boldsymbol\varepsilon}-\boldsymbol\varepsilon
=-P_{or}\boldsymbol\varepsilon+(I-P_{or})\mathbf r.
\]
Conditional covariance boundedness and $\operatorname{rank}(P_{or})\le |S_c|+d_n$ imply
\[
\frac1n\|P_{or}\boldsymbol\varepsilon\|_2^2 =O_p\{(|S_c|+d_n)/n\}=o_p(1).
\]
Assumptions~\ref{ass:A2} and \ref{ass:A4} give $\|\mathbf r\|_2^2/n=o(1)$.  The bounded cluster size and $\max_i\|\tilde{\mathbf X}_{c,i}\|_{\rm op}=O(1)$ therefore give the explicit bound
\[
\frac1n\sum_{i=1}^N\|\hat u_i-u_i\|_2^2
\le
\frac{K}{n}\|
-P_{or}\boldsymbol\varepsilon+(I-P_{or})\mathbf r
\|_2^2=o_p(1).
\]
Hence
\begin{equation}\label{eq:residual_score_consistency}
\frac1n\sum_{i=1}^N\|\hat u_i-u_i\|_2^2=o_p(1).
\end{equation}
Moreover, \eqref{eq:sandwich_lln} and $\Gamma_{c,n}\to\Gamma_c$ imply $n^{-1}\sum_i\|u_i\|_2^2=O_p(1)$.  By Cauchy--Schwarz, \eqref{eq:residual_score_consistency} and
\[
\frac1n\sum_i
\|\hat u_i\hat u_i^\top-u_iu_i^\top\|_F
\le
\left\{\frac1n\sum_i\|\hat u_i-u_i\|_2^2\right\}^{1/2}
\left\{\frac1n\sum_i(\|\hat u_i\|_2+\|u_i\|_2)^2\right\}^{1/2}
=o_p(1).
\]
Consequently,
\[
\frac1n\sum_{i=1}^N
(\hat u_i\hat u_i^\top-u_iu_i^\top)=o_p(1).
\]
Combining this display with \eqref{eq:sandwich_lln} proves $\hat\Gamma_c-\Gamma_{c,n}=o_p(1)$. Continuous mapping and $A_{c,n}\to\Sigma_c$ establish consistency of the sandwich covariance estimator.
\end{proof}

\end{document}